%% file: main.tex
\documentclass[a4paper,11pt]{article}

\usepackage[utf8]{inputenc}
\usepackage{bm,framed}

\usepackage[twoside,a4paper,margin=1in]{geometry}
\usepackage{microtype}
\usepackage{graphicx,url,amsmath,amsthm,amsfonts,amssymb,bbm,bm,enumitem}
\usepackage{hyperref}
\usepackage{mathrsfs}
\usepackage[dvips]{epsfig}
\usepackage{thm-restate}
\usepackage{float}
\usepackage{color}

\theoremstyle{plain}
\newtheorem{theorem}{Theorem}[section]
\newtheorem{proposition}[theorem]{Proposition}
\newtheorem{conjecture}[theorem]{Conjecture}
\newtheorem{corollary}[theorem]{Corollary}
\newtheorem{claim}[theorem]{Claim}
\newtheorem{lemma}[theorem]{Lemma}
\newtheorem{remark}[theorem]{Remark}

\definecolor{blueblack}{rgb}{0,0,.7}
\newcommand{\emphdef}[1]{%
  \textcolor{blueblack}{%
    \textbf{\emph{#1}}%
  }%
}

\newcommand{\eps}{\varepsilon}
\newcommand{\floor}[1]{\lfloor #1\rfloor}
\newcommand{\poly}{\text{poly}}

\newcommand{\tw}{\textup{tw}}
\newcommand{\Diag}{\textup{Diag}}
\newcommand{\REGULARCSP}{\textup{REGCSP}}
\newcommand{\CSP}{\textup{CSP}}
\newcommand{\GT}{\textup{GT}}
\newcommand{\MCone}{\textup{MC1}}
\newcommand{\MCtwo}{\textup{MC2}}
\newcommand{\MC}{\textup{MC}}
\newcommand{\CG}{\textup{CG}}
\newcommand{\out}{\textup{out}}
\newcommand{\surf}{\mathscr{S}}
\newcommand{\Input}{\textsc{Input}}
\newcommand{\Output}{\textsc{Output}}

\begin{document}

\title{Almost Tight Lower Bounds for Hard Cutting Problems in Embedded Graphs}


\author{Vincent Cohen-Addad\thanks{Sorbonne Universit\'es, UPMC Univ Paris 06, CNRS, LIP6, Paris, France, \url{vcohenad@gmail.com}} \and
Éric Colin de Verdière\thanks{LIGM, CNRS, Univ Gustave Eiffel, Marne-la-Vall\'ee, France, \url{eric.colindeverdiere@u-pem.fr}} \and
Dániel Marx\thanks{CISPA Helmholtz Center for
Information Security, Germany, \url{marx@cispa.de}} \and
Arnaud de Mesmay\thanks{LIGM, CNRS, Univ Gustave Eiffel, Marne-la-Vall\'ee,France, \url{ademesmay@gmail.com}}}

\maketitle
\begin{abstract}
We prove essentially tight lower bounds, conditionally to the Exponential Time Hypothesis, for two fundamental but seemingly very different cutting problems on surface-embedded graphs: the \textsc{Shortest Cut Graph} problem and the \textsc{Multiway Cut} problem.

A cut graph of a graph~$G$ embedded on a surface~$\surf$ is a subgraph of~$G$ whose removal from~$\surf$ leaves a disk.  We consider the problem of deciding whether an unweighted graph embedded on a surface of genus~$g$ has a cut graph of length at most a given value.  We prove a time lower bound for this problem of $n^{\Omega(g/\log g)}$ conditionally to the ETH. In other words, the first $n^{O(g)}$-time algorithm by Erickson and Har-Peled [SoCG 2002, Discr.\ Comput.\ Geom.\ 2004] is essentially optimal.  We also prove that the problem is W[1]-hard when parameterized by the genus, answering a 17-year old question of these authors.

A multiway cut of an undirected graph~$G$ with $t$ distinguished vertices, called \emph{terminals}, is a set of edges whose removal disconnects all pairs of terminals.  We consider the problem of deciding whether an unweighted graph~$G$ has a multiway cut of weight at most a given value.  We prove a time lower bound for this problem of $n^{\Omega(\sqrt{gt + g^2+t}/\log(g+t))}$, conditionally to the ETH, for any choice of the genus~$g\ge0$ of the graph and the number of terminals~$t\ge4$.  In other words, the algorithm by the second author [Algorithmica 2017] (for the more general multicut problem) is essentially optimal; this extends the lower bound by the third author [ICALP 2012] (for the planar case).

Reductions to planar problems usually involve a grid-like structure.  The main novel idea for our results is to understand what structures instead of grids are needed if we want to exploit optimally a certain value~$g$ of the genus.
\end{abstract}
\section{Introduction}\label{S:intro}

During the past decade, there has been a flurry of works investigating the complexity of solving exactly optimization problems on planar graphs, leading to what was coined as the ``square root phenomenon'' by the third author~\cite{m-srppg-13}: many problems turn out to be easier on planar graphs, and the improvement compared to the general case is captured exactly by a square root.
 For instance, problems solvable in time $2^{O(n)}$ in general graphs can be solved in time $2^{O(\sqrt{n})}$ in planar graphs, and similarly, in a parameterized setting, FPT problems admitting $2^{O(k)}n^{O(1)}$-time algorithms or W[1]-hard problems admitting $n^{O(k)}$-time algorithms can often be sped up to $2^{\widetilde{O}(\sqrt{k})}n^{O(1)}$ and $n^{\widetilde{O}(\sqrt{k})}$, respectively, when restricted to planar graphs. 
 We have many examples where matching upper bounds (algorithms) and lower bounds (complexity reductions) show that indeed  the best possible running time for the problems has this form. On the side of upper bounds, the improvement often stems from the fact that planar graphs have planar separators (and thus treewidth) of size $O(\sqrt{n})$, and the theory of bidimensionality provides an elegant framework for a similar speedup in the parameterized setting for some problems~\cite{dfht-spabgg-05}. However, in many cases these algorithms rely on highly problem-specific arguments \cite{c-mpbgg-17,mpp-spast-18,km-spktc-12,MarxP15,DBLP:conf/soda/ChitnisHM14,DBLP:conf/fsttcs/LokshtanovSW12, DBLP:conf/focs/FominLMPPS16}.
  The lower bounds are conditional to the Exponential Time Hypothesis (ETH) of Impagliazzo, Paturi, and Zane~\cite{IPZ98} and follow from careful reductions from problems displaying this phenomenon, e.g., \textsc{Planar 3-Coloring}, \textsc{$k$-Clique}, or \textsc{Grid Tiling}. We refer to the recent book~\cite{cfklm-pa-15} for precise results along these lines.

While the theme of generalizing algorithms from planar graphs to surface-embedded graphs has attracted a lot of attention, and has flourished into an established field mixing algorithmic and topological techniques (see~\cite{c-ctgs-18}), the same cannot be said at all of the lower bounds. Actually, up to our knowledge, there are very few works explicitly establishing algorithmic lower bounds based on the genus of the surfaces on which a graph is embedded, or even just hardness results when parameterized by the genus---the only ones we are aware of are the exhaustive treatise~\cite{mp-eyawk-14} of the third author and Pilipczuk on \textsc{Subgraph Isomorphism}, where some of the hardness results feature the genus of the graph, the lower bounds of Curticapean and the third author~\cite{cm-tclbcp-16} on the problem of counting perfect matchings, and the work of Chen, Kanj, Perkovi\'c, Sedgwick, and Xia~\cite{ckpsx-gcccg-07}.

In this work, we address this surprising gap by providing lower bounds conditioned on the ETH for two fundamental yet seemingly very different cutting problems on surface-embedded graphs: the \textsc{Shortest Cut Graph} problem and the \textsc{Multiway Cut} problem. In both cases, our lower bounds match the best known algorithms up to a logarithmic factor in the exponent. We believe that the tools that we develop in this paper could pave the way towards establishing lower bounds for other problems on surface-embedded graphs.

\subsection{The shortest cut graph problem}

A \emph{cut graph} of an edge-weighted graph~$G$ cellularly embedded on a surface~$\surf$ is a subgraph of~$G$ that has a unique face, which is a disk (note that such a cut graph only exists when $G$ is cellularly embedded).  Computing a shortest cut graph is a fundamental problem in algorithm design, as it is often easier to work with a planar graph than with a graph embedded on a surface of positive genus, since the
large toolbox that has been designed for planar graphs becomes available. Furthermore, making
a graph planar is useful for various purposes in computer graphics and mesh
processing, see, e.g.,~\cite{whds-reti-04}. Be it for a practical or a theoretical goal, a natural measure of the distortion induced by the cutting step is the length of the topological decomposition.

Thus, the last decade has witnessed a lot of effort on how to obtain efficient algorithms for the problems
of computing short topological decompositions, see for example the survey~\cite{c-ctgs-18}. For the shortest cut graph problem, Erickson and Har-Peled~\cite{eh-ocsd-04} showed that the problem is NP-hard when the genus is considered part of the input and gave an exact algorithm running
in time $n^{O(g)}$, where $n$ is the size of
the input graph and $g$ the genus of the surface, together with an $O(\log^2 g)$-approximation running in time
$O(g^2 n \log n)$. The first and fourth authors~\cite{cm-fptas-15} gave a $(1+\eps)$-approximation algorithm running
in time $f(\eps,g)n^3$, where $f$ is some explicit computable function.
Whether it is possible to improve upon the exact algorithm of Erickson and Har-Peled by
designing an FPT algorithm for the problem, namely an exact algorithm running in time $f(g) n^{O(1)}$, has been raised by these authors~\cite[Conclusion]{eh-ocsd-04} and has remained an open question over the last 17 years.

In this paper, we solve this question by proving that the result of Erickson and Har-Peled cannot be significantly improved.  We indeed show a lower bound of $n^{\Omega(g/\log g)}$ (for the associated decision problem, even in the unweighted case) assuming the Exponential Time Hypothesis (ETH) of Impagliazzo, Paturi, and Zane~\cite{IPZ98} (see Definition~\ref{conj:ETH}), and also prove that the problem is W[1]-hard.  More formally, we consider the following decision problem:\\
\begin{framed}
\textsc{Shortest Cut Graph}:\\
  \Input: An unweighted graph~$G$ with $n$ vertices and edges cellularly embedded on an orientable surface of genus~$g$, and an integer~$\lambda$.\\
  \Output: Whether $G$ admits a cut graph of length at most~$\lambda$.
\end{framed}
  
 \begin{theorem}\label{T:maincutgraph}
   \begin{enumerate}
      \item The \textsc{Shortest Cut Graph} problem is W[1]-hard when parameterized by~$g$.
      \item Assuming the ETH, there exists a universal constant $\alpha_{\CG}>0$ such that for any fixed integer $g\geq 0$, there is no algorithm solving all the \textsc{Shortest Cut Graph} instances of genus at most $g$ in time $O(n^{\alpha_{\CG} \cdot (g+1)/\log (g+2)})$. 
  \end{enumerate}
 \end{theorem}
In the second item, the strange-looking additive constants are just here to ensure that the theorem is still correct for the values $g=0$ and $g=1$.

\begin{remark}\label{R:constants}
Let us observe that it is sufficient to prove Theorem~\ref{T:maincutgraph} in the case where the genus $g$ is assumed to be at least some constant $\bar g$.  Indeed, assume we did so.  \textsc{Shortest Cut Graph} admits a trivial, unconditional, lower bound of $\Omega(n)$, as any algorithm has to read the input. Therefore, if we take $\alpha_\CG$ small enough so that $\alpha_\CG \cdot (g+1)/\log (g+2)<1$ for any $g\le \bar g$, then this trivial lower bound finishes the proof for the remaining constant number of values. The exact same remark will apply to the proofs of most theorems establishing lower bounds in the paper: it will suffice to prove them for the case when the parameter is larger than a certain constant, the remaining cases being trivial.
\end{remark}

\subsection{The multiway cut problem}

The second result of our paper concerns the \textsc{Multiway Cut} problem (also known as the \textsc{Multiterminal Cut} problem).  Given an edge-weighted graph~$G$ together with a subset~$T$ of $t$ vertices called terminals, a multiway cut is a set of edges whose removal disconnects all pairs of terminals.  Computing a minimum-weight multiway cut is a classic problem that generalizes the minimum $s-t$ cut problem and some closely related variants have been actively studied since as early as 1969~\cite{Hu63}. On general graphs, while the problem is polynomial-time solvable for $t=2$, it becomes NP-hard for any fixed $t\geq 3$, see~\cite{djpsy-cmc-94}. In the case of planar graphs, it remains NP-hard if $t$ is arbitrarily large, but can be solved in time $2^{O(t)}n^{O(\sqrt{t})}$, where $n$ is the number of vertices and edges of the graph~\cite{km-spktc-12}, and a lower bound of $n^{\Omega(\sqrt{t})}$ was proved (conditionally on the ETH) by the third author~\cite{m-tlbpmc-12}.  A generalization to higher-genus graphs was recently obtained by the second author~\cite{c-mpbgg-17} who devised an algorithm running in time\footnote{Note that \cite{c-mpbgg-17} states the exponent slightly imprecisely in the form $O(\sqrt{gt+g^2})$, which is not correct for, e.g., $g=0$. The current form (or, more precisely, an upper bound of $f(g,t)\cdot n^{c\sqrt{gt+g^2+t}}$ for some $c>0$) is correct for every $g\ge 0$ and $t\ge 1$.} $f(g,t)\cdot n^{O(\sqrt{gt + g^2+t})}$ 
in graphs of genus $g$, for some function~$f$.  If one allows some approximation, this can be significantly improved: three of the authors recently provided a $(1+\varepsilon)$-approximation algorithm running in time $f(\varepsilon,g,t) \cdot n \log n$~\cite{ccm-nlasm-18}.  The latter two results are actually valid for the more general \textsc{Multicut} problem, in which one looks for a minimum-weight set of edges whose removal disconnects some specified pairs of terminals (but not necessarily all of them, as opposed to \textsc{Multiway cut}).

We prove a lower bound of 
$n^{\Omega(\sqrt{gt + g^2+t}/\log(g+t))}$ for the associated decision problem, even in the unweighted case, which almost matches the aforementioned best known upper bound, and generalizes the lower bound of the third author~\cite{m-tlbpmc-12} for the planar case.  Actually, we prove a lower bound that holds for any value of the integers $g$ and $t$ as long as $t\geq 4$.
  More formally, we consider the following decision problem:\\
  \begin{framed}
  \textsc{Multiway Cut}:\\
  \Input: An unweighted graph~$G$ with $n$ vertices and edges, a set~$T$ of vertices, and an integer~$\lambda$.\\
  \Output: Whether there exists a multiway cut of~$(G,T)$ of value at most~$\lambda$.
\end{framed}
  
\begin{theorem}\label{T:mainmulticut}
  Assuming the ETH, there exists a universal constant $\alpha_{\MC}>0$ such that for any fixed choice of integers $g\geq0$ and $t\geq 4$, there is no algorithm that decides all the \textsc{Multiway Cut} instances $(G,T,\lambda)$ for which $G$ is embeddable on the orientable surface of genus~$g$ and $|T| \leq t$, in time $O(n^{\alpha_{\MC} \sqrt{gt + g^2+t}/\log(g+t)})$.
\end{theorem}

  Since \textsc{Multicut} is a generalization of \textsc{Multiway Cut}, the lower bound also holds for \textsc{Multicut}.  Note that taking $g=0$ in this theorem yields lower bounds for the \textsc{Planar Multiway Cut} problem, and recovers, up to a logarithmic factor in the exponent, the lower bounds obtained by the third author~\cite{m-tlbpmc-12} for that problem. In the opposite regime, we also prove W[1]-hardness with respect to the genus for instances with four terminals, see Proposition~\ref{P:multicut1}. We remark that $t=2$ corresponds to the minimum cut problem, which is polynomial-time solvable, so a lower bound on~$t$ is necessary. While the last remaining case, for $t=3$, is known to be NP-hard~\cite{djpsy-cmc-94}, our techniques do not seem to encompass it, and we leave its parameterized complexity with respect to the genus as an open problem.

Parameterized lower bounds in the literature often have the form ``assuming the ETH, for any function $f$, there is no $f(k)n^{o(h(k))}$ algorithm to solve problem \textsc{X}'', where $h$ is some specific dependency on the parameter. The lower bounds that we prove in Theorems~\ref{T:maincutgraph} and~\ref{T:mainmulticut} are instead of the form ``assuming the ETH, there exists a universal constant $\alpha$ such that for any fixed $k$, there is no $O(n^{\alpha h(k)})$ algorithm to solve problem \textsc{X}''. The latter lower bounds imply the former: indeed, $f(k)n^{o(h(k))}=O(n^{\alpha h(k)})$ for a fixed $k$.  Our results are stronger, concerning instances for any fixed~$k$. Moreover, lower bounds with two parameters are difficult to state with $o()$ notation. The statement of Theorem~\ref{T:mainmulticut} handles every combination of the two parameters in a completely formal way.

While Theorem~\ref{T:mainmulticut} does not use an embedded graph as an input, we can find an embedding of a graph on a surface with minimum possible genus in $2^{\poly(g)}n$ time~\cite{kmr-sltaeg-08,m-ltaega-99}. Thus, the same hardness result holds in the embedded case and the question is not about whether we are given the embedding or not.

\subsection{Main ideas of the proof}

What is a good starting problem to prove hardness results for surface-embedded graphs? For planar graphs, the \textsc{Grid Tiling} problem of the third author~\cite{m-opgas-07} has now emerged as a convenient, almost universal, tool to establish parameterized hardness results and precise lower bounds based on the ETH. A similar approach, based on constraint satisfaction problems (CSPs) on $d$-dimensional grids, was used by the third author and Sidiropoulos~\cite{ms-lbldwo-14} to obtain lower bounds for geometric problems on low-dimensional Euclidean inputs (see also~\cite{bbkmz-fetalb-18} for a similar framework for geometric intersection graphs). However, these techniques do not apply directly for the problems that we consider. Indeed, the bounds implied by these approaches are governed by the treewidths of the underlying graphs and are of the type $n^{\Omega(\sqrt{p})}$ or $n^{\Omega(p^{1-1/d})}$ respectively, where $p$ is the parameter of interest and $d$ the dimension of the grid in the latter case. In contrast, here, we are looking for bounds of the form $n^{\Omega(p/ \log p)}$ (while this is not apparent from looking at Theorem~\ref{T:mainmulticut}, this also turns out to be the main regime of interest for the \textsc{Multiway Cut} problem). 

Our first contribution, in Section~\ref{S:pivot}, is to introduce a new hard problem for embedded graphs, which is versatile enough to be used as a starting point to obtain lower bounds for both the \textsc{Shortest Cut Graph} and the \textsc{Multiway Cut} problem (and hopefully others). It is a variant of the \textsc{Grid Tiling} problem which we call \textsc{4-Regular Graph Tiling}; in a precise sense, it generalizes the \textsc{Grid Tiling} problem to allow for embedded $4$-regular graphs different from the planar grid to be used as the structure graph of the problem. We show that a CSP instance with $k$ binary constraints can be simulated by a \textsc{4-Regular Graph Tiling} instance with parameter $k$. A result of the  third author~\cite{m-cybt-10} shows that, assuming the ETH, such CSP instances cannot be solved in time $f(k)n^{\Omega(k/\log k)}$, giving a similar lower bound for \textsc{4-Regular Graph Tiling} (Theorem~\ref{T:pivot}).

We then establish in Sections \ref{S:multicut1} and~\ref{S:cutgraph} the lower bounds for the \textsc{Shortest Cut Graph} and ``one half'' of the lower bound for \textsc{Multiway Cut}, namely, for the regime where the genus dominates the number of terminals, in which case we prove a lower bound of $n^{\Omega(g/ \log g)}$.  Both reductions proceed from \textsc{4-Regular Graph Tiling} and use as a building block an intricate set of \emph{cross gadgets} originally designed by the third author~\cite{m-tlbpmc-12} for his hardness proof of the \textsc{Planar Multiway Cut} problem. While it does not come as a surprise that these gadgets are useful for \textsc{Multiway Cut} instances in the case of surface-embedded graphs, for which planar tools can often be used, it turns out that via basic planar duality, they also provide exactly the needed technical tool for establishing the hardness of \textsc{Shortest Cut Graph}.

The ``second half'' of the lower bound in Theorem~\ref{T:mainmulticut} is in the regime where the number of terminals dominates the genus, for which we establish a lower bound of $n^{\Omega(\sqrt{gt+t}/\log(g+t))}$. In Section~\ref{S:multicut2}, we use a similar strategy as before but bypass the use of the \textsc{4-Regular Graph Tiling} problem. Instead, we rely directly on the aforementioned theorem of the third author on the parameterized hardness of CSPs, which we apply not to a family of expanders, but to blow-ups of expanders, i.e., expanders where vertices are replaced by grids of a well-chosen size. This size is prescribed exactly by the tradeoff between the genus and the number of terminals, as described with the two integers $g$ and $t$ in Theorem~\ref{T:mainmulticut}. The key property of these blow-ups is that their treewidth is $\tw=\Theta(\sqrt{gt+t})$ and thus the $n^{\Omega(\tw/\log \tw)}$ lower bound on the complexity of CSPs with these blow-ups as primal graphs yields exactly the target lower bound. The reduction from CSPs to \textsc{Multiway Cut} is carried out in Proposition~\ref{P:multicut2} and also relies on cross gadgets.

The proof of Theorem~\ref{T:mainmulticut}, in Section~\ref{S:finish}, results from the two halves given by Propositions~\ref{P:multicut1} and~\ref{P:multicut2}.

Sections \ref{S:multicut1}, \ref{S:cutgraph}, and~\ref{S:multicut2} are roughly put in increasing order of technical difficulty, but they are independent; the reader interested in Theorem~\ref{T:mainmulticut} can safely skip Section~\ref{S:cutgraph}, while the reader interested in Theorem~\ref{T:maincutgraph} can safely skip Sections~\ref{S:multicut1},~\ref{S:multicut2} and~\ref{S:finish}. 

\section{Preliminaries}\label{S:prelim}

\subsection{Graphs and surfaces}\label{SS:surfaces}

For the sake of convenience in the proofs, unless otherwise noted, the graphs in this paper are not necessarily simple; they may have loops and multiple edges. However, our hardness results also hold for simple graphs, because the instances of \textsc{Shortest Cut Graph} or \textsc{Multiway Cut} can easily be made simple by subdividing edges if desired.

For extensive background on graphs on surfaces, we refer to the classic textbook of Mohar and Thomassen~\cite{mt-gs-01}. Throughout this article, we only consider \emphdef{surfaces} that are compact, connected, and orientable.  By the classification theorem of surfaces, each such surface~$\surf$ is homeomorphic to a sphere with $g$ handles attached and $b$ disks removed; $g$ is called the \emphdef{genus} of the surface and $b$ its number of \emphdef{boundaries}. A \emphdef{path}, or \emphdef{curve}, is a continuous map from $[0,1]$ to~$\surf$. A path is \emphdef{simple} if it is injective. 

An \emphdef{embedding} of $G$ on $\surf$ is a crossing-free drawing of $G$ on $\surf$, i.e., the images of the vertices are pairwise distinct and the image of each edge is a simple path intersecting the image of no other vertex or edge, except possibly at its endpoints. When embedding a graph on a surface with boundaries, we adopt the convention that while vertices can be mapped to a boundary, interiors of edges cannot.  A \emphdef{face} of the embedding is a connected component of the complement of the graph. A \emphdef{cellular embedding} is an embedding of a graph where every face is a topological disk. By a slight abuse of language, we will often identify an abstract graph with its embedding. If $G$ is a graph embedded on $\surf$, the surface obtained by \emphdef{cutting} $\surf$ along $G$ is the disjoint union of the faces of~$G$; it is an (a priori disconnected) surface with boundary.

Every graph embeddable on a surface of genus~$g$ is also embeddable on a surface of genus~$g'$, for all $g'\ge g$.  A graph embedded on a surface of minimum possible genus is cellularly embedded.  The genus of that surface is called the \emphdef{genus} of the graph; it is at most the number of edges of the graph.

To a graph cellularly embedded on a surface without boundary~$\surf$, one can naturally associate a dual graph $G^*$ embedded on~$\surf$, whose vertices are the faces of $G$ and two such vertices are connected by an edge~$e^*$ for every edge~$e$ their dual faces share; $e^*$ crosses~$e$ and no other edge of~$G$.

\subsection{Expanders, separators, and treewidth}

For a \emph{simple} graph $G$, we denote by \emphdef{$\bm{\lambda(G)}$} the second largest eigenvalue of its adjacency graph.  If $G$ is $d$-regular, it is a basic fact that $\lambda(G)\le d$.  A family of $d$-regular \emphdef{expanders} is an infinite family of $d$-regular simple graphs $G$ such that $\lambda(G)/d<c_{\textup{exp}}<1$ for some constant~$c_{\textup{exp}}$. A family $\mathcal{G}$ of graphs is \emphdef{dense} if for any $n>0$, there exists a graph in $\mathcal{G}$ with $\Theta(n)$ vertices (where the $\Theta()$ hides a universal constant).

\begin{lemma}\label{L:existence}
  There exists a dense family~$\mathcal{H}$ of bipartite four-regular expanders.
\end{lemma}

This lemma can be proved using a well-known simple probabilistic argument showing that random bipartite regular graphs are expanders, or with more intricate explicit constructions. We refer to the survey of Hoory, Linial, and Wigderson~\cite{hlw-ega-06} or the groundbreaking recent works of Marcus, Spielman, and Srivastava~\cite{mss-if1br-15}.

The \emphdef{treewidth}~$\tw(G)$ of a graph~$G$ is a parameter measuring intuitively how close it is to a tree. Since we will use this parameter in a black-box manner and not rely precisely on its definition, we do not include it here and refer to graph theory textbooks, e.g., Diestel~\cite[Chapter~12]{d-gt-00}.  We only indicate a few basic properties that we will use.  For every graph~$G$, we have $\tw(G)\le|V(G)|-1$.  Every graph containing a $\delta\times\delta$-grid as a subgraph has treewidth at least~$\delta$.

For $\alpha<1$, an \emphdef{$\bm{\alpha}$-separator} of a graph~$G$ is a subset~$C$ of vertices of~$G$ such that each connected component of~$G-C$ has a fraction at most~$\alpha$ of the vertices of~$H$.

\begin{lemma}\label{L:tw-sep}\cite[Lemma~7.19]{cfklm-pa-15}
  Let $G$ be a graph with treewidth at most~$k$.  Then $G$ has a $1/2$-separator of size at most~$k+1$.
\end{lemma}

\begin{lemma}\label{L:tw}
Let $G$ be a simple $d$-regular graph.  Then every $1/2$-separator of~$G$ has size at least $|V(G)|(d-\lambda(G))/16d$. Moreover, the treewidth of~$G$ is at least $\floor{|V(G)|\cdot(d-\lambda(G))/8d}$.
\end{lemma}
\begin{proof}
  First, we remind some standard definitions.  The edge expansion $h(G)$ of~$G$ and the (outer) vertex expansion $h_\out(G)$ of~$G$ are defined as follows, letting $n$ be the number of vertices of~$G$:
\[h(G):=\min_{\substack{A\subseteq V(G)\\1\le|A|\le n/2}}\frac{|\{uv\in E(G)\mid u\in A,v\in\bar A\}|}{|A|} \text{ , and}\]
\[h_\out(G):=\min_{\substack{A\subseteq V(G)\\1\le|A|\le n/2}}\frac{|\{v\in\bar A\mid \exists u\in A, uv\in E(G)\}|}{|A|}.\]

  First, the ``easy direction'' of the Cheeger inequality gives $(d-\lambda(G))/2\le h(G)$, see for example Hoory, Linial, and Wigderson~\cite[Theorem~2.4]{hlw-ega-06}.  Also, it is easy to see that $h(G)\le dh_\out(G)$.  So $(d-\lambda(G))/2d\le h_\out(G)$.
  
  Let us prove the bound on the size of $1/2$-separators.  Let $S$ be a $1/2$-separator for~$G$.  Assume that $|S|<n(d-\lambda(G))/16d$; in particular, $|S|\le n/8$.  Since $S$ is a $1/2$-separator, there is a set $A$ disjoint from~$S$, of size between $n/4$ and~$3n/4$, whose vertices are adjacent only to vertices in~$S\cup A$.  By replacing $A$ with $V(G)-S-A$ if necessary, we can assume that $A$ actually has size between $n/4-|S|\ge n/8$ and~$n/2$.  Thus, $h_\out(G)\le|S|/|A|\le 8|S|/n$, which together with the bound of the previous paragraph implies the result.

  To prove the lower bound on the treewidth, note that additionally to the above bound on~$h_\out$, we have $\floor{n/4\cdot h_\out(G)}\le\tw(G)$, by Grohe and the third author~\cite[Proposition~1]{grohe2009tree} (in which we choose $\alpha=1/2$).
\end{proof}

\subsection{Exponential Time Hypothesis}

Our lower bounds are conditioned on the Exponential Time Hypothesis (ETH), which was conjectured in~\cite{IPZ98} and is stated below.  \emphdef{3SAT} denotes the Boolean satisfiability problem in which instances are presented in conjunctive normal form with at most three literals per clause.

\begin{conjecture}[Exponential Time Hypothesis~\cite{IPZ98}]
  \label{conj:ETH}
There exists a positive real value $s > 0$ such that 3SAT, parameterized by $n$, has no $2^{sn}(n+m)^{O(1)}$-time algorithm, where $n$ denotes the number of variables
and $m$ denotes the number of clauses.
\end{conjecture}

We refer to the survey~\cite{lms-lbbeth-13} for background and discussion of this conjecture. Informally speaking, Conjecture~\ref{conj:ETH} states that there is no algorithm for 3SAT that is subexponential in the number of \textit{variables} of the formula.
  The Sparsification Lemma of 
Impagliazzo, Paturi, and Zane~\cite[Corollary~1]{ipz-wphse-01} shows that the ETH is equivalent to saying that there is no algorithm subexponential in the \textit{length} of the formula.
\begin{corollary}\label{cor:eth-sparse}
  If the ETH holds, then there exists a positive real value $s > 0$ such that 3SAT has no $2^{s(n+m)}(n+m)^{O(1)}$-time algorithm, where $n$ denotes the number of variables and $m$ denotes the number of clauses.
\end{corollary}

\subsection{Constraint satisfaction problems}

A \emphdef{binary constraint satisfaction problem} is a triple $(V,D,C)$ where 
\begin{itemize}
\item $V$ is a set of \emphdef{variables},
\item $D$ is a \emphdef{domain} of values,
\item $C$ is a set of \emphdef{constraints}, each of which is a triple of the form $\langle u,v,R \rangle$, where $(u,v)$ is a pair of variables called the \emphdef{scope}, and $R$ is a subset of $D^2$ called the \emphdef{relation}.
\end{itemize}

All the constraint satisfaction problems (CSPs) in this paper will be binary, and thus we will omit the adjective binary. A solution to a constraint satisfaction problem instance is a function $f: V \rightarrow D$ such that for each constraint $\langle u,v,R\rangle$, the pair $(f(u),f(v))$ is a member of $R$. An algorithm \emphdef{decides} a CSP instance $I$ if it outputs true if and only if that instance admits a solution.

The \emphdef{primal graph} of a CSP instance $I=(V,D,C)$ is a graph with vertex set $V$ such that distinct vertices $u,v \in V$ are adjacent if and only if there is a constraint whose scope contains both $u$ and $v$.

The starting points for the reductions in this paper are the following two theorems, which state in a precise sense that the treewidth of the primal graph of a binary CSP establishes a lower bound on the best algorithm to decide it.
\begin{theorem}[{\cite{gss-wecqt-01,g-chcsps-07}}]\label{T:W1}
Let $\mathcal{G}$ be an arbitrary class of graphs with unbounded treewidth.  Let us consider the problem of deciding the binary CSP instances whose primal graph, $G$, lies in~$\mathcal G$. This problem is W[1]-hard parameterized by the treewidth of the primal graph.
\end{theorem}

\begin{theorem}[{\cite{m-cybt-10}}]\label{T:beattreewidth}
    Assuming the ETH, there exists a universal constant $\alpha_{\CSP}$ such that for any fixed primal graph $G$ such that $\tw(G)\geq 2$, there is no algorithm deciding the binary CSP instances whose primal graph is $G$ in time $O(|D|^{\alpha_{\CSP} \cdot \tw(G)/\log \tw(G)})$.
\end{theorem}

Theorem~\ref{T:W1} is due to Grohe, Schwentick, and Segoufin~\cite{gss-wecqt-01} (see also Grohe~\cite{g-chcsps-07}). Theorem~\ref{T:beattreewidth} follows from the work of the third author~\cite{m-cybt-10}. Since this statement is stronger than the main theorem of~\cite{m-cybt-10} (which assumes the existence of an algorithm solving binary CSP instances for a class of graphs with unbounded treewidth), we explain how to prove it in the rest of this subsection. This reformulation could be useful also for other problems where lower bounds of this form are proved; it seems to be especially important for the clean handling of lower bounds with respect to two parameters in the exponent.

We first need to recall some definitions and results from \cite{m-cybt-10}. 
A graph $H$ is a \emphdef{minor} of $G$ if $H$ can be obtained from $G$ by
a sequence of vertex deletions, edge deletions, and edge
contractions. Equivalently, a graph $H$ is a minor of $G$ if there is a {\em minor mapping} from $H$ to $G$, which is a function $\psi: V(H)\to 2^{V(G)}$ satisfying the following properties: (1) $\psi(v)$ is a connected vertex set in $G$ for every $v\in V(H)$, (2) $\psi(u)\cap \psi(v)=\emptyset$ for every $u\neq v$, and (3) if $uv\in E(H)$, then there is an edge of $G$ intersecting both $\psi(u)$ and $\psi(v)$. 
 Given a graph $G$ and an integer $q$, we denote by
$G^{(q)}$ the graph obtained by replacing every vertex with a clique
of size $q$ and replacing every edge with a complete bipartite graph
on $q+q$ vertices. The main combinatorial result of \cite{m-cybt-10} is the following embedding theorem:

\begin{theorem}[{\cite[Theorem~3.1]{m-cybt-10}}]\label{th:embedding}
There are computable functions $f_1(G)$, $f_2(G)$, and a universal constant $c$ such that for every $k\ge 2$, if $G$ is a graph with $\tw(G)\ge k$ and $H$ is a graph with $|E(H)|=m \ge f_1(G)$ and no isolated vertices, then $H$ is a minor of $G^{(q)}$ for $q=\lceil cm\log k/k \rceil$. Furthermore, such a minor mapping can be found in time $f_2(G)m^{O(1)}$.
\end{theorem}

We will also need the following (fairly straightforward) reductions from \cite{m-cybt-10}:

\begin{lemma}\label{lem:sat2csp}
  Given an instance of 3SAT with $n$ variables and $m$ clauses, it is
  possible to construct in polynomial time an equivalent CSP instance
  with $n+m$ variables, $3m$ binary constraints, and domain size $3$.
\end{lemma}
\begin{lemma}\label{lem:cspreductionminor}
  Assume that $G_1$ is a minor of $G_2$. Given a binary CSP instance
$I_1$   with primal graph $G_1$ and a minor mapping from $G_1$
to $G_2$, it is possible to
  construct in polynomial time an
  equivalent instance $I_2$ with primal graph $G_2$ and the same domain.
\end{lemma}

\begin{lemma}\label{lem:cspreductionblowup}
  Given a binary CSP instance $I_1=(V_1,D_1,C_1)$
  with primal graph $G^{(q)}$ (where $G$ has no isolated vertices), it is possible to
  construct (in time polynomial in the size of the {\em output}) an
  equivalent instance $I_2=(V_2,D_2,C_2)$ with primal graph
  $G$ and $|D_2|=|D_1|^q$.
\end{lemma}

With these tools at hand, we are ready to prove Theorem~\ref{T:beattreewidth}.

\begin{proof}[Proof of Theorem~\ref{T:beattreewidth}]
Assume that the ETH holds.  By Corollary~\ref{cor:eth-sparse}, let $s>0$ be a universal constant such that 3SAT has no $2^{s(n+m)}(n+m)^{O(1)}$-time algorithm, where $n$ denotes the number of variables and $m$ denotes the number of clauses.

Let $c\ge1$ be a universal constant satisfying Theorem~\ref{th:embedding} and let us define $\alpha_\CSP=s/(6c\log_2 3)$.  Let $G$ be a graph with treewidth $k\ge2$.  Throughout this proof, we consider that $G$ (and thus~$k$) are fixed, in the sense that the $O()$ notation can hide factors depending on $G$ and~$k$.  In contrast, the notation ``$\poly$'' denotes a fixed polynomial, independent of $G$ and~$k$.  Suppose that an algorithm~$\mathbb{A}_G$ decides the binary CSP instances whose primal graph is~$G$ in time $O(|D|^{\alpha_\CSP\cdot k/\log k})$.

Consider an instance $\Psi$ of 3SAT with $n$ variables and $m$ clauses.  Using Lemma~\ref{lem:sat2csp}, we construct in $\poly(n+m)$ time an instance $I_1$ of CSP equivalent to~$\Psi$ with at most $n+m$ variables, at most~$3m$ binary constraints, and domain size 3.  Let $H$ be the primal graph of $I_1$; without loss of generality, it does not have any isolated vertex. If $3m<\max\{f_1(G),k\}$, then we can solve $I_1$, and thus~$\Psi$, in $O(1)$ time (hiding a factor depending on~$G$ and~$k$). Otherwise, by Theorem~\ref{th:embedding}, graph $H$ is a minor of $G^{(q)}$ for $q=\lceil c(3m)\log k/k \rceil \le 2c(3m)\log k/k$ (the inequality uses the facts that $c\ge1$ and $3m\ge k\ge2$), and the minor mapping can be found in time $f_2(G)\poly(m)$. Then, by Lemma~\ref{lem:cspreductionminor}, $I_1$~can be turned in time $\poly(n+m)$ into an instance $I_2$ with primal graph $G^{(q)}$ and domain size 3, which, by Lemma~\ref{lem:cspreductionblowup}, can be turned in time $O(\poly(3^{2q}))$ into an instance $I_3$ with primal graph $G$ and domain size $3^q$.  Using these reductions and algorithm~$\mathbb A_G$ to solve~$I_3$, we can solve~$\Psi$ in time $\poly(n+m)+2^{\beta q}+O((3^q)^{\alpha_\CSP\cdot k/\log k})$, for some universal constant $\beta>0$ that does not depend on~$k$.  By definition of~$q$ and $\alpha_\CSP$, this is \[\poly(n+m)+2^{\gamma m\log k/k}+O(2^{sm})\] for some universal constant $\gamma>0$ that does not depend on~$k$.

Let $\bar k\ge2$ be a universal constant such that $\gamma\log k/k\le s$ for every $k\ge\bar k$; then the total running time becomes $O(2^{sm}\poly(n+m))$. Therefore, we have proved that if there is a graph $G$ with $\tw(G)\ge \bar k$ and an algorithm~$\mathbb A_G$ deciding the binary CSP instances with primal graph~$G$ in time $O(|D|^{\alpha_\CSP\cdot k/\log k})$, then we can solve any 3SAT instance in time $2^{s(n+m)}\poly(n+m)$, which, by our choice of~$s$, implies that the ETH does not hold. By Remark~\ref{R:constants}, this suffices to conclude the proof.
\end{proof}

We will rely in particular on the following corollary of Theorems~\ref{T:W1} and~\ref{T:beattreewidth}:

\begin{corollary}\label{C:beattreewidth}
  \begin{enumerate}
    \item Deciding the binary CSP instances whose primal graph has at
    most $k$ vertices, is four-regular and bipartite is W[1]-hard
    when parameterized by $k$.
  \item Assuming the ETH, there exists a universal constant $\alpha_{\CSP}$ such that for any fixed $k \geq 2$, there is no algorithm deciding the binary CSP instances whose primal graph has at most $k$ vertices, is four-regular and bipartite in time $O(|D|^{\alpha_{\REGULARCSP} \cdot k/\log k})$.
  \end{enumerate}
\end{corollary}

\begin{proof}
For the first item, we apply Theorem~\ref{T:W1} to the infinite family $\mathcal{H}$
of four-regular bipartite expanders output by Lemma~\ref{L:existence}.
Then, Lemma~\ref{L:tw} implies that for each such graph,
the treewidth is linear in the number of vertices, and so we conclude
that the problem is W[1]-hard when parameterized by $k$.

For the second item, by Lemma~\ref{L:existence}, there are universal constants~$c>0$, $\bar k>0$, and~$c_{\textup{exp}}<1$ such that, for any $k\ge \bar k$, there exists a four-regular bipartite expander~$P$ such that $\lambda(P)/4<c_{\textup{exp}}$, and with at least $ck$ but at most $k$ vertices.  By Lemma~\ref{L:tw}, for some universal constant~$c'>0$, the treewidth of~$P$ is at least~$c'k$, which we can assume to be at least two (up to increasing~$\bar k$).

We choose $\alpha_{\REGULARCSP}$ so that $\alpha_{\REGULARCSP} \leq \alpha_{\CSP}/c'$. Let us assume that there exists an algorithm as described by the corollary.  Assume first that $k\ge\bar k$, and let $P$ be obtained as in the previous paragraph.  The assumed algorithm would in particular decide the binary CSP instances whose primal graph is $P$ in time $O(|D|^{\alpha_{\REGULARCSP} \cdot k/\log k}) = O(|D|^{\alpha_{\REGULARCSP}/c' \cdot \tw(P)/\log \tw(P)}) = O(|D|^{\alpha_{\CSP} \cdot \tw(P)/\log \tw(P)}$ (for fixed~$P$, with $\tw(P)\ge2$). By Theorem~\ref{T:beattreewidth}, this would contradict the ETH.
 
Thus the second part of the corollary is proved for $k\ge \bar k$. The other cases are trivial if $\alpha_{\REGULARCSP}$ is small enough, as in Remark~\ref{R:constants} (because if $k\ge2$, there indeed exist some binary CSP instances whose primal graph has the form indicated in the statement of the corollary).
\end{proof}

\subsection{Cross gadgets}

We rely extensively on the following intricate family of gadgets introduced by the third author in his proof of hardness of \textsc{Planar Multiway Cut}~\cite{m-tlbpmc-12}, which we call \emphdef{cross gadgets}; see Figure~\ref{F:cross1}, left.  Let $\Delta$ be an integer.   The gadgets always have the form of a planar graph $G_S$ embedded on a disk, with $4\Delta+8$ distinguished vertices on its boundary, which are, in clockwise order, denoted by
\[UL, u_1, \ldots, u_{\Delta+1}, UR, r_1, \ldots , r_{\Delta+1}, DR, d_{\Delta+1}, \ldots d_1, DL, \ell_{\Delta+1}, \ldots ,\ell_1.\] 
The embedding is chosen so that the boundary of the disk intersects the graph precisely in this set of distinguished vertices; the interior of the edges lie in the interior of the disk.  We consider the vertices $UL, UR, DR$, and $DL$ as terminals in that gadget, and thus a multiway cut $M$ of the gadget is a subset of the edges of $G_S$ such that $G_S \setminus M$ has at least four components, and each of the terminals is in a distinct component. We say that a multiway cut $M$ of the gadget \emphdef{represents} the pair $(i,j) \in [\Delta]^2$ (where, as usual, $[\Delta]$ denotes the set $\{1,\ldots,\Delta\}$) if $G_S \setminus M$ has exactly four components that partition the distinguished vertices into the following classes:
\begin{align*}
  \{UL,u_1, \ldots, u_j, \ell_1, \ldots, \ell_i \}& \quad &\{UR, u_{j+1}, \ldots _{\Delta+1}, r_1, \ldots ,r_i\} \\
  \{DL, d_1, \ldots d_j , \ell_{i+1}, \ldots , \ell_{\Delta+1}\} & \quad & \{DR,d_{j+1}, \ldots, d_{\Delta+1}, r_{i+1}, \ldots r_{\Delta+1}\}
\end{align*}

We remark that, as in the original article~\cite{m-tlbpmc-12}, the notation $(i,j)$ is in matrix form (i.e., (row, column), with $(1,1)$ being in the corner~$UL$). We will use the same convention throughout this paper, especially in Section~\ref{S:pivot}.

The properties that we require are summarized in the following lemma:

\begin{figure}
    \centering
    \def\svgwidth{.85\textwidth}
    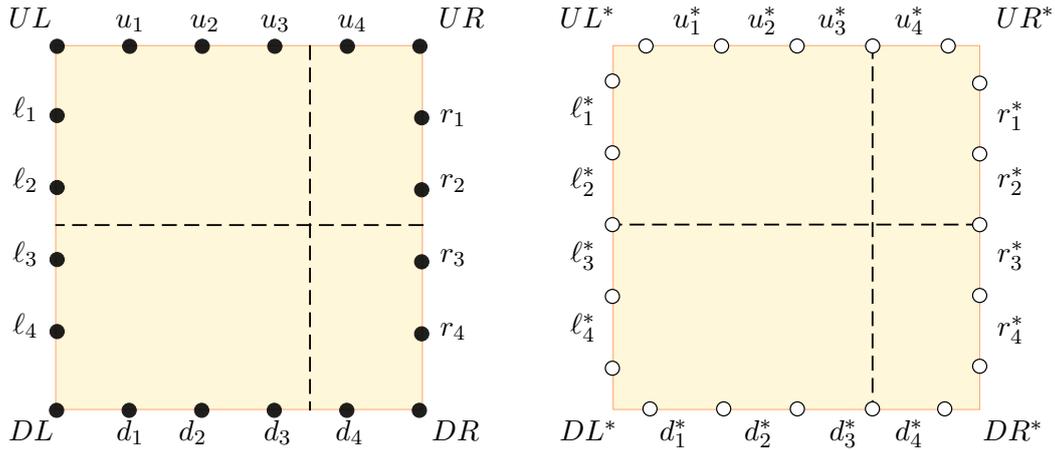
    \caption{Left: a cross gadget $G_S$ for $\Delta=3$. The dashed line indicates a multiway cut that represents the pair $(2,3)$. Right: a dual cross gadget $G_S^*$ for $\Delta=3$. The dashed lined is a dual multiway cut that represents the pair $(2,3)$.}
    \label{F:cross1}
\end{figure}

\begin{lemma}[{\cite[Lemma~2]{m-tlbpmc-12}}]\label{L:marx:gadgets}
  Given a subset $S \subseteq [\Delta]^2$,
  we can construct in polynomial time a planar gadget $G_S$ with $\poly(\Delta)$ unweighted edges and vertices, and an integer $D_1$ such that the following properties hold:
  \begin{enumerate}\renewcommand{\theenumi}{\roman{enumi}}
  \item For every $(i,j) \in S$, the gadget $G_S$ has a multiway cut of weight $D_1$ representing $(i,j)$.
  \item Every multiway cut of $G_S$ has weight at least $D_1$.
  \item If a multiway cut of $G_S$ has weight $D_1$, then it represents some $(i,j) \in S$.
  \end{enumerate}
\end{lemma}

Note that in~\cite{m-tlbpmc-12}, the third author uses weights to define the gadgets, but as he explains at the end of the introduction, the weights are polynomially large integers and thus can be emulated with parallel unweighted edges.

In the following, we will also use the dual of the graph $G_S$ as one of our gadgets, yielding a \emphdef{dual cross gadget} $G_S^*$ (see Figure~\ref{F:cross1}). Its properties mirror exactly the ones of cross gadgets in a dual setting. In the dual setting, the gadget $G_S^*$ still has the form of a planar graph embedded on a disk $D$, with $4\Delta+8$ distinguished faces incident to its boundary, which are, in clockwise order, denoted by
\[UL^*, u_1^*, \ldots, u_{\Delta+1}^*, UR^*, r_1^*, \ldots , r_{\Delta+1}^*, DR^*, d_{\Delta+1}^*, \ldots d_1^*, DL^*, \ell_{\Delta+1}^*, \ldots ,\ell_1^*.\] 

In $G^*_S$, the vertices dual to boundary faces of $G_S$ lie on the boundary on the disk instead of the interior, see Figure~\ref{F:cross1}, right.  As above, the boundary of the disk intersects the graph $G^*_S$ precisely in the distinguished vertices.

A dual multiway cut is a set of edges $M^*$ of $G_S^*$ such that cutting the disk $D$ along $M^*$ yields at least four connected components, and the four terminal faces end up in distinct components. We say that a dual multiway cut $M^*$ \emphdef{represents} a pair $(i,j) \in [\Delta]^2$ if cutting the disk $D$ along $M^*$ yields exactly four connected components that partition the distinguished faces into the following classes:

\begin{align*}
  \{UL^*,u_1^*, \ldots, u_j^*, \ell_1^*, \ldots, \ell_i^* \}& \quad &\{UR^*, u_{j+1}^*, \ldots u_{\Delta+1}^*, r_1^*, \ldots ,r_i^*\} \\
  \{DL^*, d_1^*, \ldots d_j^* , \ell_{i+1}^*, \ldots , \ell_{\Delta+1}^*\} & \quad & \{DR^*,d_{j+1}^*, \ldots, d_{\Delta+1}^*, r_{i+1}^*, \ldots r_{\Delta+1}^*\}
\end{align*}

For convenience, we restate the content of Lemma~\ref{L:marx:gadgets} in the dual setting in a separate lemma.  By duality, its proof follows directly from Lemma~\ref{L:marx:gadgets}.

\begin{lemma}[{\cite[Dual version of Lemma~2]{m-tlbpmc-12}}]\label{L:marx:gadgetsdual}
  Given a subset $S \subseteq [\Delta]^2$,
  we can construct in polynomial time a planar gadget $G_S^*$ with $\poly(\Delta)$ unweighted edges and vertices, and an integer $D_1$ such that the following properties hold:
  \begin{enumerate}\renewcommand{\theenumi}{\roman{enumi}}
  \item For every $(i,j) \in S$, the gadget $G_S^*$ has a dual multiway cut of weight $D_1$ representing $(i,j)$.
  \item Every dual multiway cut of $G_S^*$ has weight at least $D_1$.
  \item If a dual multiway cut of $G_S^*$ has weight $D_1$, then it represents some $(i,j) \in S$.
  \end{enumerate}
\end{lemma}

\section{The \texorpdfstring{4}{}-regular graph tiling problem}\label{S:pivot}

We introduce the problem $\textsc{4-Regular Graph Tiling}$ which will be used as a basis to prove the reductions involved in Theorems~\ref{T:maincutgraph} and~\ref{T:mainmulticut}.

\begin{framed}
\textsc{4-Regular Graph Tiling}

\textbf{Input:} Positive integers $k$, $\Delta$; a four-regular graph~$\Gamma$ on~$k$ vertices where the edges are labeled by $U,D,L,R$ in a way that each vertex is incident to exactly one of each label; for each vertex~$v$, a set~$S_v \subseteq [\Delta] \times [\Delta]$.

\textbf{Output:} For each vertex~$v$, a value $s_v \in S_v$ such that if $s_v=(i,j)$, 
\begin{enumerate}
\item the first coordinates of $s_{L(v)}$ and $s_{R(v)}$ are both~$i$, and
\item the second coordinates of $s_{U(v)}$ and $s_{D(v)}$ are both~$j$,
\end{enumerate}
where $U(v)$, $D(v)$, $L(v)$, and $R(v)$ denote the vertex of the graph $\Gamma$ connected to~$v$ via an edge labeled respectively by $U$, $D$, $L$, and $R$.
\end{framed}

We call the two conditions above the \emphdef{compatibility conditions} of the \textsc{4-Regular Graph Tiling} instance. The graph in the input is allowed to have parallel edges. It is easy to see that the \textsc{Grid Tiling} problem~\cite{m-opgas-07} is a special case of \textsc{4-Regular Graph Tiling}. 

In this section, we prove a larger lower bound for this more general problem: We prove a $\Delta^{\Omega(k/\log k)}$ lower bound, conditionally to the ETH, for \textsc{4-Regular Graph Tiling}, even when the problem is restricted to bipartite instances and when fixing $k$.  We also show that it is W[1]-hard when parameterized by the integer~$k$ (even for bipartite instances).   Precisely:
\begin{theorem}\label{T:pivot}
\begin{enumerate}
\item The \textsc{4-Regular Graph Tiling} problem restricted to instances whose underlying graph is bipartite is W[1]-hard parameterized by the integer~$k$.
  \item Assuming the ETH, there exists a universal constant $\alpha_{\GT}$ such that for any fixed integer $k\geq 2$, there is no algorithm that decides all the \textsc{4-Regular Graph Tiling} instances whose underlying graph is bipartite and has at most $k$ vertices, in time $O(\Delta^{\alpha_{\GT} \cdot k / \log k})$.
\end{enumerate}
\end{theorem}

The analogous result for \textsc{Grid Tiling} by the third author~\cite{m-opgas-07} embeds the $k$-\textsc{Clique} problem in a $k\times k$ grid. Here we start from a hardness result for 4-regular binary CSPs that follows from Theorem~\ref{T:beattreewidth} and directly represent the problem as a \textsc{4-Regular Graph Tiling} instance by locally replacing each variable and each binary constraint in an appropriate way.

\begin{proof}
\begin{figure}
    \centering
    \def\svgwidth{\textwidth}
    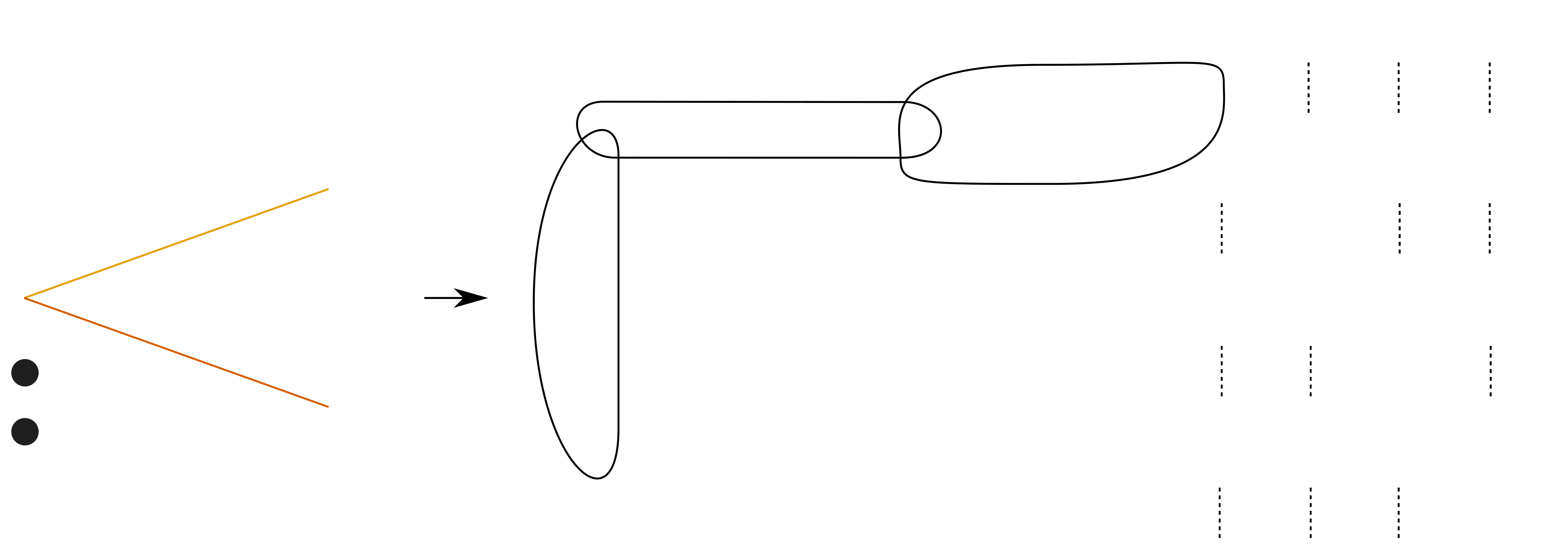
    \caption{The reduction in the proof of Theorem~\ref{T:pivot}. The bipartition on both sides are represented by hollow/solid vertices. The colors represent the $4$-coloring of the edges, and the labels of the edges are suggested by their orientation, i.e., edges entering vertices vertically are labeled $U$ or $D$, while edges entering vertices horizontally are labeled $L$ or $R$.}
    \label{F:GraphTiling}
\end{figure}

In the proof, we will use the well-known fact that a $d$-regular bipartite graph~$G$ can be properly edge-colored with $d$ colors.  This is proved by induction on~$d$:  The case $d=0$ is trivial; in general, take a perfect matching of~$G$, which exists by Hall's marriage theorem; color the edges with color~$d$; the subgraph of~$G$ made of the uncolored edges satisfies the induction hypothesis with $d-1$, so it admits a proper edge-coloring with $d-1$ colors; thus $G$ has a proper edge-coloring with $d$ colors.  This also implies that computing such a proper edge-coloring takes polynomial time.

The proof of the theorem proceeds by a reduction from the binary CSP instances involved in Theorem~\ref{T:W1} and Corollary~\ref{C:beattreewidth}.  Starting from a binary CSP instance $I=(V,D,C)$ whose primal graph is $P$, a $4$-regular bipartite graph, we define an instance of \textsc{4-Regular Graph Tiling}, $(k,\Delta,\Gamma,\{S_i\})$, in the following way (see Figure~\ref{F:GraphTiling}):

\begin{enumerate}
\item We set $\Delta=|D|$ and $k=6|V|$.
\item We find a proper edge coloring of $P$ with $4$ colors, as indicated above.
\item Denoting by $V_1$ and $V_2$ the two subsets of vertices of $P$ corresponding to the bipartition of $P$, for each vertex $u$ of $V_1$, we create four vertices $u_1,u_2,u_3,u_4$ in $\Gamma$ which we connect in a cycle in this order using two $U$ and two $D$ edges (e.g., $u_1u_2$ and $u_3u_4$ are $U$ edges and $u_2u_3$ and $u_4u_1$ are $D$ edges). Similarly, for each vertex $v$ of $V_2$, we create four vertices $v_1,v_2,v_3,v_4$ in $\Gamma$ which we connect in a cycle in this order using two $R$ and two $L$ edges.
\item For each edge $e=uv$ labeled with a color $i$, where $u \in V_1$ and $v \in V_2$, we create one vertex $v_e$ in $\Gamma$, which is connected to $u_i$ via two edges, one labeled $R$ and one labeled $L$, and to $v_i$ via two edges, one labeled $U$ and one labeled $D$. \item For each vertex $u_i$ or $v_i$ of $\Gamma$ coming from a vertex of $P$, the corresponding subset $S_{u_i}$ or $S_{v_i}$ is set to be $\Diag([\Delta]):=\{(x,x)\mid x\in[\Delta]\}$.
\item For each vertex $v_e$ of $\Gamma$ coming from an edge $e=uv$ of $P$, where $u\in V_1$ and~$v\in V_2$, the corresponding subset $S_e$ is set to be the relation corresponding to $e$.
\end{enumerate}

We first prove that the graph $\Gamma$ is bipartite.  For this purpose, let $W_1$ be the set of vertices $\{u_i\mid u\in V_1,i\in\{1,3\}\}\cup\{v_i\mid v\in V_2,i\in\{1,3\}\}$, together with the vertices~$v_e$ such that $e$ is labeled by an even color.  Let $W_2$ be the vertices of~$\Gamma$ not in~$W_1$.  By construction, each edge in~$\Gamma$ connects a vertex with $W_1$ with a vertex in~$W_2$, so $\Gamma$~is indeed bipartite.

We claim that this instance of \textsc{4-Regular Graph Tiling} is satisfiable if and only if $I$ is satisfiable. Indeed, if $I$ is satisfiable, the truth assignment $f$ for $I$ can be used to find the values for the $s_i$ in the following way.  If $u\in V_1$ and~$i\in[4]$, the value of $s_{u_i}$ is chosen to be $(f(u),f(u))$.  Similarly, if $v\in V_2$ and~$i\in[4]$, the value of $s_{v_i}$ is chosen to be $(f(v),f(v))$.  For a vertex $v_e$ of $\Gamma$ coming from an edge $e=uv$ of $P$ where $u\in V_1$ and $v\in V_2$, the value of $s_{e}$ can be chosen to be $(f(u),f(v))$. The compatibility conditions are trivially fulfilled. In the other direction, the values $s_{v_i}$ for the four vertices of $\Gamma$ coming from a vertex $v$ of $P$ are identical and of the form $(x,x)$. Choosing $x$ as the truth assignment for $v$ in $I$ yields a solution to the CSP $I$.

We thus have a reduction from binary CSP, restricted to instances~$I$ whose primal graph~$P$ has $|V|$ vertices, is four-regular and is bipartite, to instances of \textsc{4-Regular Graph Tiling} on a bipartite graph~$\Gamma$ with $6|V|$ vertices.  This reduction takes time $O((|V||D|)^d)$ for some universal constant~$d\ge1$.  Combined with the
first item of
Corollary~\ref{C:beattreewidth} this proves the first item of the theorem.

For the second item, let $\alpha_\GT=\alpha_\REGULARCSP/6$ and let us assume that we have an algorithm~$\mathbb A$ deciding the $4$-\textsc{Regular Graph Tiling} bipartite instances with at most $k$ vertices in time $O(\Delta^{\alpha_\GT\cdot k/ \log k})$ (where, in the $O()$ notation, we consider $k$ to be fixed). For any four-regular and bipartite graph with at most $k/6$ vertices, the graph obtained from it in the reduction above has at most $k$ vertices. Using this reduction and algorithm~$\mathbb A$, we could therefore decide any binary CSP instance $(V,D,C)$ whose primal graph is four-regular, bipartite, and has at most $k/6 \geq 2$ vertices in time $O(|D|^d+|D|^{\alpha_\GT \cdot k/\log k})$ (for fixed~$k$). This is $O(|D|^{\alpha_\REGULARCSP \cdot (k/6)/\log (k/6)})$ if $k$ is larger than some universal constant $\bar k$ (that depends only on $d$ and $\alpha_\REGULARCSP$), and thus, by Corollary~\ref{C:beattreewidth}, the ETH does not hold.

This means that that the second item is proved for $k$ at least $\bar k$. The remaining cases follow from the trivial linear lower bound, as in Remark~\ref{R:constants} (note that, because $k\ge2$, there are such instances of \textsc{4-Regular Graph Tiling}).
\end{proof}

\textbf{Remark}: It might seem more natural to use a definition of \textsc{4-Regular Graph Tiling} where \textit{half-edges} are labeled by $U,D,L$ and $R$, so that every edge contains either $U$ and $D$, or $L$ and $R$ labels. This fits more the intuition that the top side of a vertex should be attached to the bottom side of the next vertex. It follows from roughly the same proof that the same hardness result also holds for that variant. However, it seems that both the bipartiteness and the unusual labeling are required for the reduction in Section~\ref{S:multicut1}.

\section{Multiway cut with four terminals}\label{S:multicut1}

In this section, we prove the following proposition, which will yield Theorem~\ref{T:mainmulticut} in the regime where the genus dominates the number of terminals.  The other case will be handled in Section~\ref{S:multicut2}.

\begin{proposition}\label{P:multicut1}
\begin{enumerate}
    \item The \textsc{Multiway Cut} problem when restricted to instances $(G,T,\lambda)$ in which $|T|=4$ and $G$ is a graph embeddable on the surface of genus $g$ is W[1]-hard parameterized by $g$.
    \item Assuming the ETH, there exists a universal constant $\alpha_{\MCone}$ such that for any fixed integer $g \geq 0$, there is no algorithm that decides all the \textsc{Multiway Cut} instances $(G,T,\lambda)$ for which $|T|=4$ and $G$ has $n$ vertices and edges and is embeddable on the surface of genus $g$, in time $O(n^{\alpha_{\MCone} \cdot (g+1)/\log (g+2)})$.
\end{enumerate}
\end{proposition}
\begin{proof}
The idea is to reduce \textsc{4-Regular Graph Tiling} instances of Theorem~\ref{T:pivot} to the instances of \textsc{Multiway Cut} specified by the proposition.  Consider an instance $(k,\Delta,\Gamma,\{S_i\})$ of \textsc{4-Regular Graph Tiling} where the underlying graph~$\Gamma$ is bipartite.  In polynomial time, we transform it into an equivalent instance $(G,T,\lambda)$ of \textsc{Multiway Cut} as follows.
\begin{enumerate}
\item For each vertex $v$ of $\Gamma$, we create a cross gadget $G_S(v)$ such that the set $S\subseteq [\Delta]^2$ is chosen to be equal to~$S_v$.
\item For each edge $e=uv$ of $\Gamma$ labeled $U$, we identify the vertices of the $U$ side of the cross gadget $G_S(v)$ to the corresponding vertices of the $U$ side of the cross gadget $G_S(u)$. Similarly for the edges labeled $D$, $R$, and $L$ for which the vertices on the $D$, $R$, and~$L$ sides, respectively, are identified. Note that only vertices, and not edges, are identified.
\item The four corner vertices $UL, UR, DR$, and $DL$ of all the cross gadgets are identified in four vertices $UL, UR, DR$, and $DL$, where the four terminals are placed.
\item We let $\lambda:=kD_1$ (where $D_1$ is the integer from Lemma~\ref{L:marx:gadgets}).
\end{enumerate}
Note that since the sides are consistently matched in this last step, the four terminals remain distinct after this identification.  (Of course, this identification may create loops and multiple edges.)

We claim that this instance admits a multiway cut of weight at most $\lambda$ if and only if the \textsc{4-Regular Graph Tiling} instance is satisfiable. Assume first that the \textsc{4-Regular Graph Tiling} instance is satisfiable.  For each vertex~$v$ of $\Gamma$, one can use the value $s_v$ to choose, using Lemma~\ref{L:marx:gadgets}(1), a multiway cut in $G_S(v)$ representing $s_v$.  Let $M$ be the union of all these sets of edges.  We claim that $M$ is a multiway cut separating the four terminals in $G$. Indeed, after removing the multiway cuts, the four terminals lie in four different components in each of the cross gadgets.  It suffices to prove that it remains the case after identifying the four sides.  To see this, consider two cross gadgets that have two sides identified, and let~$w$ be a vertex on that common side.  Then, by the compatibility conditions in the definition of \textsc{4-Regular Graph Tiling}, $w$ is connected, in the first gadget, to a terminal ($UL$, $UR$, $DR$, or $DL$) if and only if it is connected to the corresponding terminal in the second gadget.  Thus, as desired, $M$ is a multiway cut separating the four terminals in~$G$.  Moreover, it has weight at most $kD_1$, since it is the union of $k$ edge sets of weight at most $D_1$.

For the other direction, we first observe that if the instance admits a multiway cut of weight at most $kD_1$, then each of the cross gadgets $G_S$ must admit a multiway cut (otherwise the four terminals would not be disconnected).  By Lemma~\ref{L:marx:gadgets}(2), each of these $k$~multiway cuts has weight exactly~$D_1$. Therefore, by Lemma~\ref{L:marx:gadgets}(3), each of them represents some $(i,j) \in S$, which will be used as the value $s_v$ for the \textsc{4-Regular Graph Tiling} instance. Furthermore, we claim that the multiway cuts need to match along identified sides, by which we mean that the two following conditions are satisfied: (1) If a multiway cut represents the pair $(i,j)$, then a multiway cut in a cross gadget adjacent along an edge labeled $U$ or $D$ needs to represent a pair $(k,j)$ for some $k\in [\Delta]$.  (2) Similarly, a multiway cut in a cross gadget adjacent along an edge labeled $R$ or $L$ needs to represent a pair $(i,\ell)$ for some $\ell\in [\Delta]$. Indeed (see Figure~\ref{F:Newclaim}), if, say, a multiway cut representing the pair $(i,j)$ is connected along an edge labeled $R$ to a multiway cut representing the pair $(i',\ell)$ for $i'>i$, then vertex $(i',\Delta)$, common to both gadgets, is connected to~$UR$ by a path inside the first gadget and to~$DR$ by a path inside the second gadget, contradicting the fact that we have a multiway cut. Therefore, the compatibility conditions of the \textsc{4-Regular Graph Tiling} instance are satisfied.

\begin{figure}
\centering
\def\svgwidth{7cm}
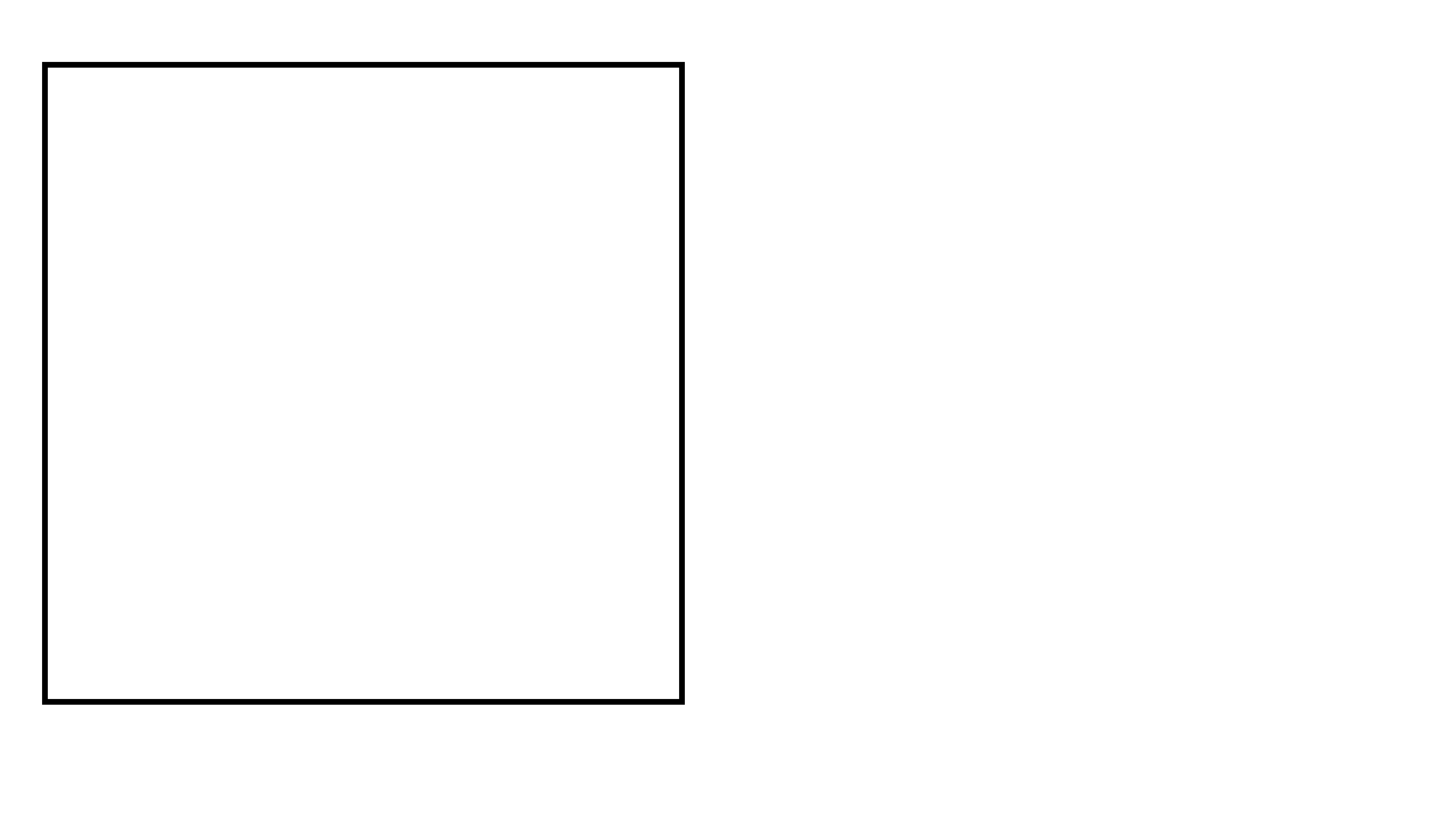
\caption{If the multiway cuts do not match (here represented by their duals), they do not separate the terminals.}
\label{F:Newclaim}
\end{figure}
\begin{figure}
    \centering
    \includegraphics[width=.8\linewidth]{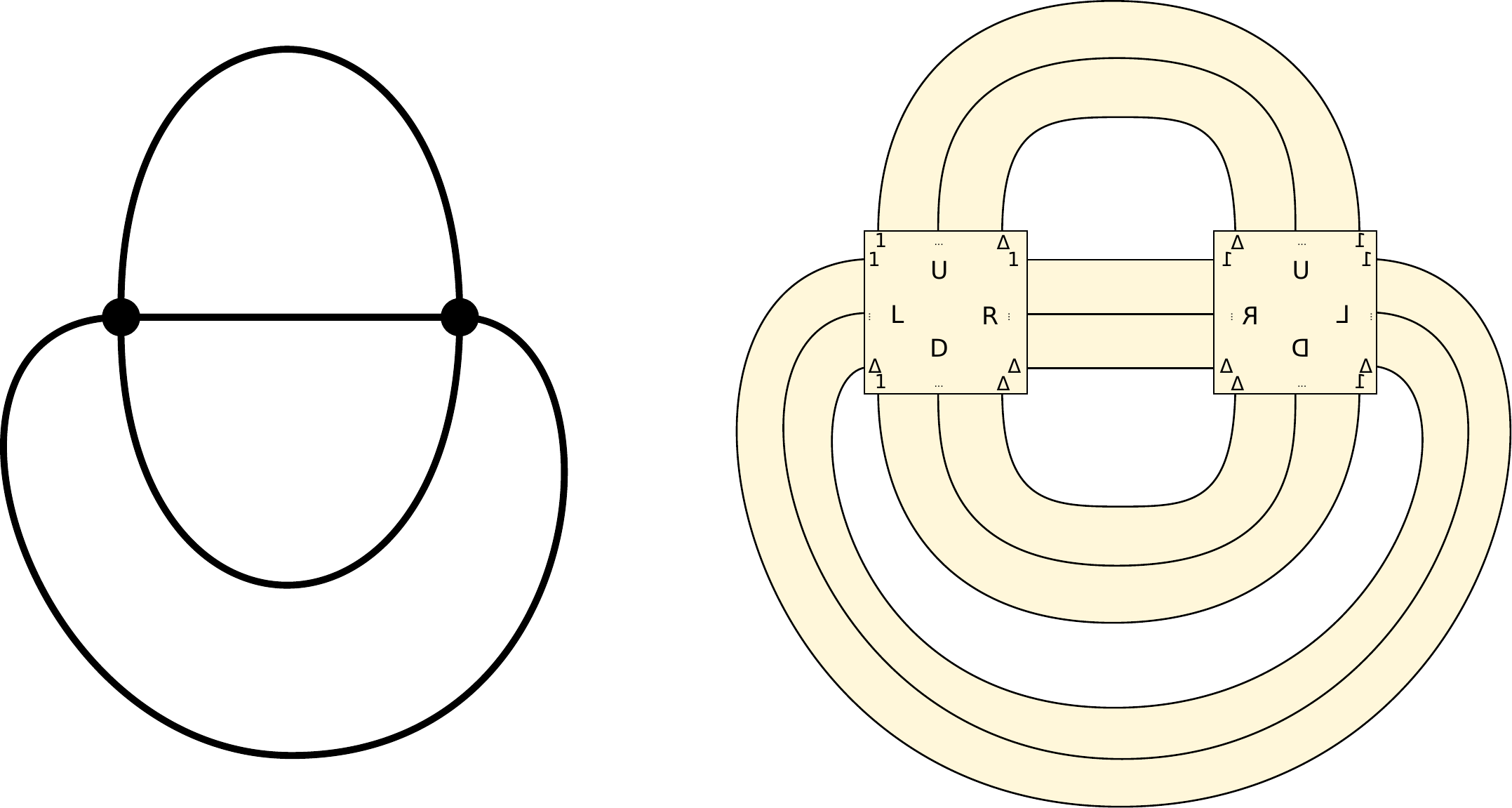}
    \caption{Left: A bipartite four-valent graph~$\Gamma$ with two vertices.  Right: The construction of the embedding  of~$G$.  The orientation of each gadget of~$G$, corresponding to a vertex~$v$, is chosen according to the side of the bipartition vertex~$v$ lies in.  This allows to connect pairs of vertices on the boundary of each gadget with the same indices.}
    \label{F:orient4}
\end{figure}
  We will prove the following claim.
\begin{claim}\label{C:genus}
The genus of the graph $G$ is $O(k)$.
\end{claim}
\begin{proof}
We prove here that the genus of the graph $G$ is $O(k)$.  For this, the fact that $\Gamma$ is bipartite turns out to be crucial.  For Step~1 above, let us embed the cross gadgets corresponding to~$\Gamma$ in the plane.  Let $V_1\cup V_2$ be the bipartition of the vertices of~$\Gamma$.  We embed the cross gadgets corresponding to~$V_1$ in the plane with the natural orientation ($U$, $R$, $D$, $L$ in clockwise order), and the cross gadgets corresponding to~$V_2$ with the opposite orientation.  Second, let us connect the vertices on the sides of the cross gadgets as in Step~2 above; but for now, just for clarity of exposition, instead of identifying pairs of vertices, let us connect each pair by a new edge.  We can add these $n$~new edges corresponding to a single edge of~$\Gamma$ by putting them on a ribbon connecting the sides of the cross gadgets (see Figure~\ref{F:orient4}).  We emphasize that, because of the orientation chosen to embed the gadgets corresponding to $V_1$ and~$V_2$, the ribbons are drawn ``flat'' in the plane (though possibly with some overlapping between them), and the vertices in one cross gadget are connected to the corresponding vertices in the other cross gadget (for example, in the case of an edge labeled $R$, vertex $r_i$ in the first gadget is connected to vertex $r_i$ in the second gadget).  Thus, since we started with a graph embedded on the plane and added at most $2k$ ``flat'' ribbons (because $\Gamma$ is four-regular), we obtain a graph embedded on an orientable surface with genus at most~$2k$ (without boundary, after attaching disks to each boundary component).  We now contract every newly added edge, which can only decrease the genus.  For Step~3 above, the graph~$G$ is obtained by identifying four groups of at most $k$~vertices of the previous graph (the terminals $UL$ of all cross gadgets, and similarly for $UR$, $DL$, and $DR$) into four vertices; these vertex identifications increase the genus by $O(k)$.  (To see this, we can for example add $O(k)$ edges to connect in a linear way all the vertices to be identified, which increases the genus by $O(k)$, and then contract these new edges.)  This proves that $G$ is embeddable on a surface of genus $O(k)$.
\renewcommand{\qedsymbol}{$\lrcorner$}
\end{proof}

To summarize: For some universal constants~$c,d\ge 1$, we can transform in time $O((k\Delta)^d)$ any instance $(k,\Delta,\Gamma,\{S_i\})$ of \textsc{4-Regular Graph Tiling} where $\Gamma$~is bipartite into an equivalent instance of \textsc{Multiway Cut} with four terminals and whose graph has $O((k\Delta)^d)$ vertices and edges and is embeddable on a surface of genus at most~$ck$.  Combined with Theorem~\ref{T:pivot}(1), this proves the first item.

Let us now consider the second item.  Let $\alpha_\MCone=\alpha_\GT/c(d+1)$ and assume that for some fixed $g$ there is an algorithm~$\mathbb A$ that decides all the \textsc{Multiway Cut} instances $(G,T,\lambda)$ for which $G$ has $n$ vertices and edges and is embeddable on the surface of genus~$g$ and $|T|=4$ in time $O(n^{\alpha_\MCone \cdot (g+1)/\log (g+2)})$.  Let $k':=\lfloor (g+2)/c\rfloor$ be fixed, and consider an instance $(k,\Delta,\Gamma,\{S_i\})$ of \textsc{4-Regular Graph Tiling} whose underlying graph is bipartite and has $k\le k'$ vertices.  Using algorithm~$\mathbb A$ and the above reduction, we can decide this instance in time $O(\Delta^d)+(O(\Delta^d))^{\alpha_\MCone \cdot c(k'+1)/\log (ck')}$ (for fixed~$k'$).  If $g$ is larger than some universal constant $\bar g$ (and thus~$k'$ is also large enough), then this is $O(\Delta^{\alpha_\GT k'/\log k'})$.  Theorem~\ref{T:pivot}(2) then implies that the ETH does not hold.

This means that if $g\ge\bar g$, then the second item is proved. The remaining cases follow from the trivial linear lower bound, as in Remark~\ref{R:constants}.
\end{proof}

\section{Shortest cut graph}\label{S:cutgraph}

In this section, we prove Theorem~\ref{T:maincutgraph} on the hardness of the \textsc{Shortest Cut Graph} problem.

\begin{proof}[Proof of Theorem~\ref{T:maincutgraph}]
  The idea is to reduce \textsc{4-Regular Graph Tiling} instances of Theorem~\ref{T:pivot} to instances of \textsc{Shortest Cut Graph}.

  Let $\Gamma$ be a bipartite four-regular graph with $k$ vertices.  From~$\Gamma$, we build an orientable surface without boundary~$\surf$ as follows (see Figure~\ref{F:cutgraph}): We build one cylindrical tube for each edge of~$\Gamma$ and one sphere minus four disks for each vertex of~$\Gamma$, attaching them in the natural way to obtain an orientable surface.

  We claim that this surface has genus~$k+1$.  To see this, let $T$ be a spanning tree of~$\Gamma$; it has $k-1$ edges, and $E(\Gamma)\setminus T$ has $k+1$~edges.  For each edge $e\in E(\Gamma)\setminus T$ in turn, let us remove from~$\surf$ a cylindrical tube corresponding to~$e$, and attach a disk to each of the two resulting boundary components.  Each of these $k+1$ operations decreases the genus by one (as can be seen formally using Euler's formula).  We eventually obtain a sphere.  This proves the claim.

  Moreover, the graph~$\Gamma$ is naturally embedded in~$\surf$, though not cellularly.  In order to have a cellular embedding, and actually a cut graph, we transform~$\Gamma$ as follows. Let $\Gamma'$ be the graph obtained from~$\Gamma$ by subdividing each edge not in~$T$ into two edges, and adding a loop in the middle vertex.  Now, embed~$\Gamma'$ into~$\surf$ in the natural way: Starting from the embedding of~$\Gamma$ into~$\surf$, put each middle vertex on the corresponding cylindrical tube of~$\surf$, and make the corresponding loop go around the tube.  The resulting embedding of~$\Gamma'$ is a cut graph of~$\surf$.  Indeed, it has a single face, because we only add loops in the middle of edges not in the spanning tree~$T$; moreover, it has $2k+1$ vertices and $4k+2$ edges (being four-regular), so its unique face is a disk, by Euler's formula.
\begin{figure}
    \centering
    \includegraphics[width=.8\linewidth]{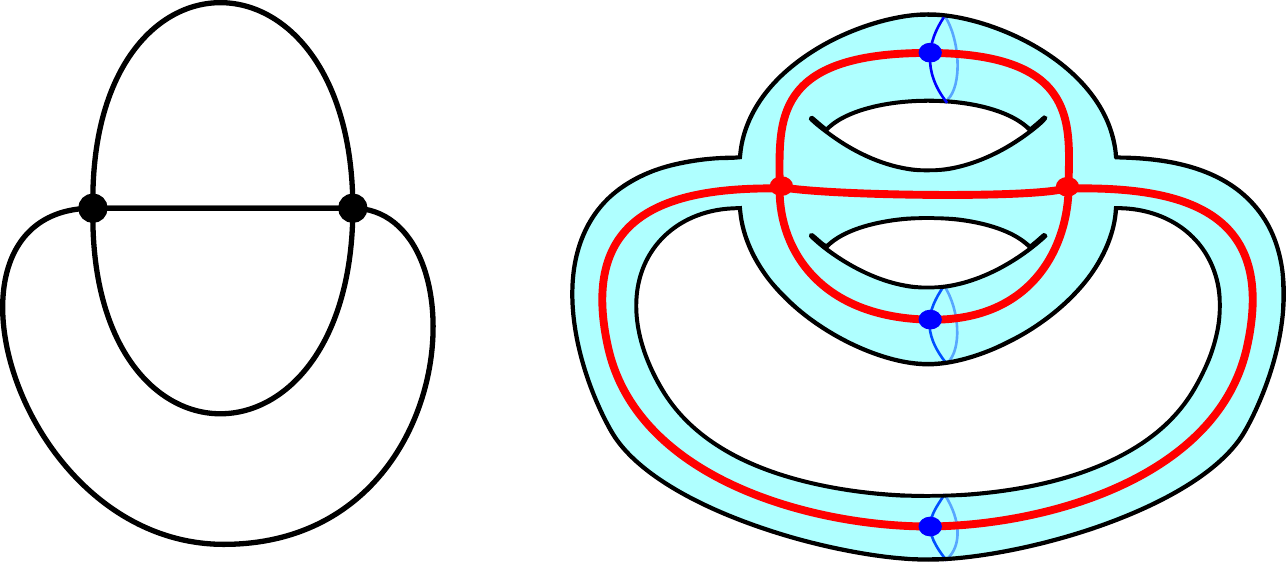}
    \caption{Left: A bipartite four-valent graph~$\Gamma$ with two vertices.  Right: The resulting graph~$\Gamma'$.  The graph~$\Gamma$ is in thick lines, and the edges forming the complement of a spanning tree are split with a new vertex, to which a loop (in thin lines) is attached.}
    \label{F:cutgraph}
\end{figure}

Let $V_1\cup V_2$ be the bipartition of the vertices of~$\Gamma$.   We note that the above construction is possible while enforcing an arbitrary cyclic ordering of the edges incident to each vertex of~$\Gamma$; we do it in a way that the cyclic ordering of the edges around each vertex in~$V_1$ is the standard one ($U$, $R$, $D$, $L$ in clockwise order), while the cyclic ordering around each vertex in~$V_2$ is reversed ($U$, $L$, $D$, $R$ in clockwise order).  We now build a graph~$G$ embedded on the same surface~$\surf$, obtained by replacing each vertex of~$\Gamma'$ with a dual cross gadget and by identifying vertices on the corresponding sides of adjacent gadgets.  In detail:
\begin{enumerate}
\item For each vertex $v$ of $\Gamma$, we create a dual cross gadget $G_S^*(v)$.  To define this gadget, we only have to specify the corresponding integer $\Delta$ together with the subset~$S$ of~$[\Delta]^2$; we choose $\Delta:=n$ and $S:=S_v$.  We embed that dual cross gadget with the same orientation as the corresponding vertex of~$\Gamma'$.
\item For each edge $e=uv$ of $T$, we identify the vertices (not the edges) on the side of $G_S^*(u)$ corresponding to the label of $e$ to the vertices on the same side of $G_S^*(v)$.  By the choice of the rotation systems, and for the same reason as in Figure~\ref{F:orient4}, this identifies the vertices in the gadget associated with~$u$ to the corresponding vertices in the gadget associated with~$v$; for example, if the label of edge~$e$ is~$R$, the vertex~$r_i$ of the first gadget is associated to vertex~$r_i$ of the second gadget.

\item For an edge $e=uv$ of~$\Gamma$ not in $T$, we use another dual cross gadget $G_S^*(e)$, for which we choose $S$ to be the unconstrained relation $S=[n]^2$.  We put that gadget on the vertex of~$\Gamma'$ that splits edge~$e$.  We identify the vertices on the side of $G_S^*(u)$ corresponding to the label of $e$ to the vertices of the same side of $G_S^*(e)$, and similarly the vertices on the side of $G_S^*(v)$ corresponding to the label of $e$ to the vertices on the opposite side of $G_S^*(e)$. The two sets of vertices on the opposite sides of $G_S^*(e)$ which are not yet identified are identified to each other.
\end{enumerate}

As a result of the identifications, some corners of the dual cross gadgets might become identified. Unlike in the reductions for \textsc{Multiway Cut}, this is not a problem here. The following claim shows that the reduction works as expected.

\begin{claim}\label{C:cutgraph}
  The \textsc{4-Regular Graph Tiling} instance on $\Gamma$ is satisfiable if and only if the embedded graph $G$ admits a cut graph $C$ of weight at most $(2k+1)D_1$.
\end{claim}
\begin{proof}
(Recall that $\Gamma'$ has $2k+1$ vertices.)  Let us first assume that the \textsc{4-Regular Graph Tiling} instance is satisfiable.  For each vertex~$v$ of $\Gamma$, one can use the value $s_v$ to choose, using Lemma~\ref{L:marx:gadgetsdual}(1), a dual multiway cut in $G^*_S(v)$ representing $s_v$. For $s_v=(i,j)$, such a dual multiway cut contains a path connecting the vertex between the $i$th and $(i+1)$st face of $G^*_S(v)$ on the side $L$ to the vertex between $i$th and $(i+1)$st face of $G^*_S(v)$ on the side $R$, and likewise there is a path connecting the $U$ and $D$ side between the $j$th and $(j+1)$st faces. Since the $(s_v)$ satisy the compatibility conditions, such paths in two adjacent dual cross gadgets must match on the common boundary. Thus the union of the edges of these dual multiway cuts contains a cut graph of~$\surf$; indeed, it contains a subdivision of the graph $\Gamma'$ described above (possibly after transforming a degree-four vertex into two degree-three vertices connected by an edge), which is a cut graph.  Furthermore, this union has weight~$(2k+1)D_1$, since each of the dual multiway cuts has weight~$D_1$.  Thus, some subgraph of it, of weight at most~$(2k+1)D_1$, is a cut graph.

The reverse direction requires more effort. First, let us call a cut graph \emph{reduced} if it has no degree-one vertex; $\Gamma'$ is clearly reduced.  We will use the following fact: Every simple closed curve~$\gamma$ crossing some reduced cut graph exactly once is non-contractible.  Indeed, if $\gamma$ is contractible, it bounds a disk; since the cut graph intersects the boundary of the disk exactly once, the part of the cut graph inside the disk is a tree, and thus has at least two degree-one vertices, one of which is not on~$\gamma$; this is a contradiction.

Assume that the embedded graph~$G$ admits a cut graph~$C$ of weight at most $(2k+1)D_1$.  Note that the edge set of~$G$ is the disjoint union of the edges of the dual cross gadgets, and one can naturally talk about the restriction of (the set of edges of) $C$ to a given dual cross gadget (before the identification of the vertices of the dual cross gadgets).  We first prove that the restriction of~$C$ to each of the dual cross gadgets~$G_S^*$ contains a dual multiway cut.  Indeed, assume that there is a path~$p$, in the disk corresponding to gadget~$G^*_S$, that connects two terminal faces and does not cross~$C$.  Let $q$ be a path with the same endpoints as~$p$ and that does not enter the gadgets; such a path exists because the union of the gadgets is a thickened version of the graph~$\Gamma'$, which is a cut graph, and which thus has a single face, and because the terminal faces are on the boundary of the dual cross gadget.  The closed curve~$\gamma$ that is the concatenation of $p$ and~$q$ is contractible, because it does not cross~$C$.  On the other hand, we can slightly modify~$\Gamma'$ in the disk of the gadget~$G_S^*$ to obtain a reduced cut graph that crosses~$p$, and thus also~$\gamma$, exactly once (which, by the previous paragraph, implies that $\gamma$ is contractible, and thus the contradiction).  To prove this fact, there are two cases.  If $p$ connects two ``neighboring'' terminals in~$G^*_S$, e.g., $UL$ and~$DL$, then we locally homotope the part of~$\Gamma'$ inside the disk of~$G^*_S$ to the~$R$ side of that disk, except for the edge that crosses the~$L$ side, and we draw that latter edge in a way that it crosses~$p$ exactly once.  If $p$ connects two ``opposite'' terminals in~$G^*_S$, e.g., $UL$ and~$DR$, then we first replace the four-valent vertex of~$\Gamma'$ in the gadget with two three-valent vertices $v_1$ and~$v_2$, in a way that $v_1$ is connected to the $L$ and $D$ sides and to~$v_2$, and similarly $v_2$ is connected to the $U$ and~$R$ sides and to~$v_1$.  We put $v_1$ in the $DL$ corner, $v_2$ in the $UR$ corner, and connect them by an edge crossing~$p$ exactly once.  Thus, in both cases, we have a (clearly reduced) cut graph that crosses~$p$, and thus~$\gamma$, exactly once, as desired.

Thus, the restriction of $C$ to each dual cross gadget $G_S^*(u)$ or $G_S^*(e)$ is a dual multiway cut representing some $(i,j)\in S$, by Lemma~\ref{L:marx:gadgets}(2) and (3).  Furthermore, the dual multiway cuts must match on the boundaries, i.e., if a dual multiway cut represents the pair $(i,j)$, the multiway cut in a cross gadget adjacent along an edge labeled $U$ or $D$ needs to represent a pair $(k,j)$ for some $k\in [\Delta]$, and similarly the dual multiway cut in a cross gadget adjacent along an edge labeled $R$ or $L$ needs to represent a pair $(i,\ell)$ for some $\ell\in [\Delta]$. Indeed, otherwise, there is a path connecting two terminal faces, as pictured in Figure~\ref{F:Newclaim}, and by a similar argument as above, this path can be completed into a non-contractible cycle not crossing $C$, which is a contradiction.

Since the restriction of $C$ to each dual cross gadget is a dual multiway cut that has to represent some $(i,j) \in S$, we use $(i,j)$ as the value of~$s_v$ for the \textsc{4-Regular Graph Tiling} instance.  The compatibility conditions follow from the fact that the dual multiway cuts match on the boundaries; thus, the \textsc{4-Regular Graph Tiling} instance is satisfiable.
\renewcommand{\qedsymbol}{$\lrcorner$}
\end{proof}

To summarize: For some universal constant~$d\ge 1$, we can transform in time $(k\Delta)^d$ any instance $(k,\Delta,\Gamma,\{S_i\})$ of \textsc{4-Regular Graph Tiling} where~$\Gamma$ is bipartite into an equivalent instance of~\textsc{Shortest Cut Graph} whose graph has at $O((k\Delta)^d)$ vertices and edges, embedded on a surface of genus~$k+1$. Combined with Theorem~\ref{T:pivot}(1), this proves the first item of the theorem.

Let us now consider the second item.  Let $\alpha_\CG=\alpha_\GT/(d+1)$ and assume that, for some fixed $g$, there is an algorithm~$\mathbb A$ that decides all the \textsc{Shortest Cut Graph} instances embedded on a surface of genus at most~$g$ in time $O(n^{\alpha_\CG\cdot (g+1)/\log (g+2)})$.  Let $k'=g-1$, and consider an instance $(k,\Delta,\Gamma,\{S_i\})$ of \textsc{4-Regular Graph Tiling} whose underlying graph is bipartite and has $k\le k'$ vertices.  Using algorithm~$\mathbb A$ and the above reduction, we can decide this instance in time $O(\Delta^d)+(\Delta^d)^{\alpha_{\CG} \cdot (k'+2)/\log(k'+3)}$ (for fixed~$g$, and thus fixed~$k'$).  If $g$ is larger than a certain universal constant $\bar g$, then this is $O(\Delta^{\alpha_\GT k'/\log k'})$. Theorem~\ref{T:pivot}(2) then implies that the ETH does not hold.

This means that there exists an integer~$\bar g$ such that, if $g\ge\bar g$, then the theorem is proved. The remaining cases follow from the trivial linear lower bound, as in Remark~\ref{R:constants}.
\end{proof}

\section{Multiway cut with a large number of terminals}\label{S:multicut2}

The goal of this section is to prove the following proposition, which will yield Theorem~\ref{T:mainmulticut} in the regime where the number of terminals dominates the genus.

\begin{proposition}\label{P:multicut2}  
  Assuming the ETH, there exists a universal constant $\alpha_{\MCtwo}$ such that for any fixed choice of integers $g\geq0$ and $t\geq 48(g+1)$, there is no algorithm that decides all the \textsc{Multiway Cut} instances $(G,T,\lambda)$ for which $G$ is embeddable on the surface of genus~$g$ and such that $|T| \leq t$, in time $O(n^{\alpha_{\MCtwo} \sqrt{gt+t}/\log(g+t)})$ .
\end{proposition}

\begin{figure}
    \centering
    \def\svgwidth{\textwidth}
    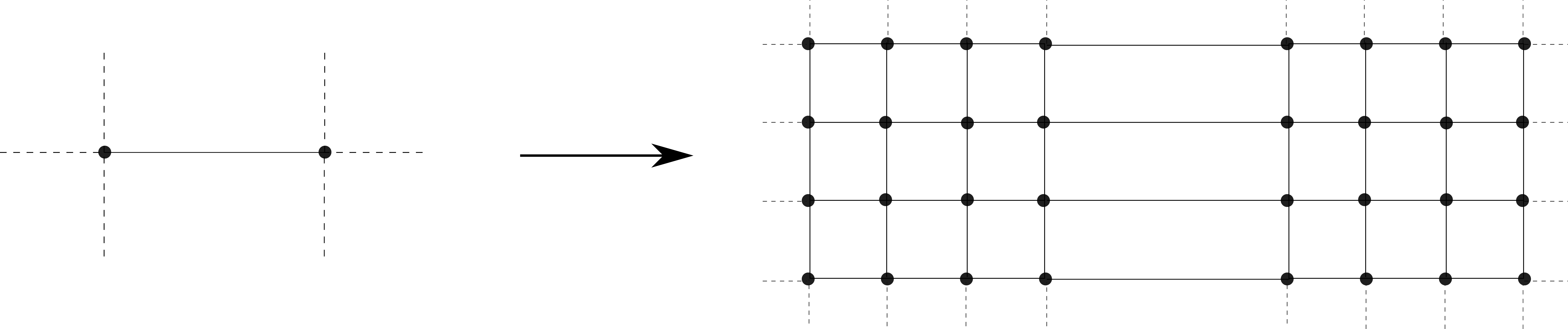
    \caption{The construction of $G_\delta$ for $\delta=4$.}
    \label{F:blowup}
\end{figure}

In order to prove this proposition, we will blow up expander graphs by replacing their vertices by grids, a construction similar to the one in Gilbert, Hutchinson and Tarjan~\cite{ght-stgbg-84}. Given a four-regular graph cellularly embedded on a surface $\surf$, we define $G^\delta$ in the following way (see Figure~\ref{F:blowup}).  We replace each vertex of $G$ with a $(\delta\times\delta)$-grid.  The boundary of the grid is thus made of $4\delta-4$ vertices, which we divide into four contiguous segments of $\delta$ vertices starting from a corner, in clockwise order along the boundary of the grid (corners appear in two segments). For each edge $e$ of $G$ between two adjacent vertices $u$ and $v$, we blow up $e$ in the following way: we place an edge between the $i$th vertex of a segment of $u$ and the $(\delta-i+1)$st vertex of a segment of $v$. The segments are selected in such a way that the cyclic ordering of the four blown-up half-edges around a grid replacing a vertex $u$ corresponds to the cyclic ordering of the four half-edges around $u$ that is prescribed by the cellular embedding. This new graph $G^\delta$ is also cellularly embedded on~$\surf$, since the embedding of $G$ can be ``thickened'' to an embedding of $G^\delta$. Moreover, $G^\delta$ clearly has $\delta^2\cdot|V(G)|$ vertices.

The properties of $G^\delta$ that we are interested in are summarized in the following lemma.
\begin{lemma}\label{L:ght}
  Let $G$ be a four-regular graph cellularly embedded on a surface~$\surf$, such that $\lambda(G)/4<c<1$ for some universal constant $c$, and let~$\delta$ be a positive integer.  The graph $G^\delta$ can also be cellularly embedded on~$\surf$, has $\delta^2\cdot|V(G)|$ vertices and treewidth $\Omega(\delta\cdot |V(G)|)$.
\end{lemma}

To prove the previous lemma, we will use the following auxiliary, probably folklore, result.
\begin{lemma}\label{L:grid}
  Let $H$ be the $(\delta_1\times\delta_2)$-grid where $\delta_1\le\delta_2$, and let $\alpha<1$ be given.  There exists $\beta>0$, depending only on~$\alpha$ and~$\delta_2/\delta_1$, such that any $\alpha$-separator of~$H$ has more than~$\beta\delta_1$ vertices.
\end{lemma}
\begin{proof}
  So let $H$ have $\delta_1$ rows and~$\delta_2$ columns.
  Let $\beta$ be chosen such that $0<\beta<\min\{1,(1-\alpha)\delta_2/\delta_1\}$.  Let $C$ be a subset of vertices of~$H$ of size at most $\beta\delta_1$.  Thus, 
  at least one row and $\delta_2-\beta\delta_1$ columns do not intersect~$C$, and there is a connected component of $H-C$ of size at least $\delta_1(\delta_2-\beta\delta_1)>\alpha\delta_1\delta_2$, and $C$ is not an $\alpha$-separator.
\end{proof}

\begin{proof}[Proof of Lemma~\ref{L:ght}]
  Following the construction, and by Lemma~\ref{L:tw}, there just remains to prove that $\tw(G^\delta)=\Omega(\delta\cdot \tw(G))$.  The approach mimicks the lower bound proof of Gilbert, Hutchinson, and Tarjan~\cite{ght-stgbg-84} and adapts it to the case of four-regular graphs. Their initial proof is for three-regular graphs.  Let $G':=G^\delta$.  For each vertex~$v$ of the original graph~$G$, we denote by $gr(u)$ the associated $(\delta\times\delta)$-grid in~$G'$.

  Most of the proof consists in proving that if $G'$ has a $1/4$-separator of size~$k$, then $G$ has a $1/2$-separator of size~$O(k/\delta)$.  For this purpose, we first note, by Lemma~\ref{L:grid}, the existence of a universal constant $\beta>0$ such that:
  \begin{itemize}
  \item for every vertex~$u$ of~$G$, every $1/2$-separator of~$gr(u)$ has more than $\beta\delta$ vertices, and
  \item for every edge~$uv$ of~$G$, every $3/4$-separator of~$G'[gr(u)\cup gr(v)]$ has more than $2\beta\delta$ vertices.
  \end{itemize}
  Let $C'$ be a $1/4$-separator for~$G'$.  Let $C$ be the set of vertices~$u$ of~$G$ such that $gr(u)$ contains more than $\beta\delta$ vertices of~$C'$.  We will prove that $C$ is a $1/2$-separator of $G$ of size $O(|C'|/\delta)$.

  Let $u$ be a vertex of~$G-C$.  By definition of $\beta$ and~$C$, we have that $gr(u)\cap C'$ is not a $1/2$-separator of~$gr(u)$, and thus the strict majority of the vertices of~$gr(u)$ lie in a single connected component of~$G'-C'$.  Thus, let $A'_1,\ldots,A'_r$ be the connected components of~$G'-C'$; for $i=1,\ldots,r$, let $A_i$ be the set of vertices~$u$ of~$G-C$ such that the strict majority of the vertices of $gr(u)$ lie in~$A'_i$.  We have that $C,A_1,\ldots,A_r$ form a partition of the vertices of~$G$.

  We first argue that $C$ separates all pairs $A_i,A_j$, $i \neq j$.  Assume to the contrary that there are adjacent vertices $u$ and~$v$ of~$G$ such that $u\in A_i$ and~$v\in A_j$.  By definition of $A_i$ and~$A_j$, neither $gr(u)$ nor~$gr(v)$ contains more than $\beta\delta$ vertices of $C'$.  Thus $G'[gr(u) \cup gr(v)]$ contains at most $2\beta\delta$ vertices of $C'$.  Thus, by definition of~$\beta$, a fraction larger than $3/4$ of its vertices of~$G'[gr(u) \cup gr(v)]$, namely more than $3\delta^2/2$ vertices, are in a single component of~$G'-C'$.  It follows that both $gr(u)$ and~$gr(v)$ have more than $\delta^2/2$ vertices in that component, so they are both in the same set $A'_i$, a contradiction.

  We then show that each $A_i$ is small.  Let $n$ be the number of vertices of~$G$.  Recall that (i) $|V(A'_i)|\le\delta^2n/4$, because $C'$ is a $1/4$-separator, and (ii) $|V(A'_i)|>|V(A_i)|\cdot\delta^2/2$, by definition of~$A_i$.  These two inequalities imply $|V(A_i)|\le n/2$.

  Hence $C$ is a $1/2$-separator.  Finally, by definition of~$C$, we have $|C'|>\beta\delta|C|$.  Thus $|C|=O(|C'|/\delta)$.

  We have proved that if $G'$ has a $1/4$-separator of size $k$, then $G$ has a $1/2$-separator of size $O(k/\delta)$.  On the other hand, since $\lambda(G)/4<c<1$, by Lemma~\ref{L:tw}, there exists a constant~$c'>0$ such that $G$ has no $1/2$-separator of size at most $c'|V(G)|$.  Therefore, there exists a constant $c''>0$ such that $G'$ does not have a $1/4$-separator of size at most $c'' \delta |V(G)|$.  This implies that there exists a constant $c'''>0$ such that $G'$ has treewidth at least $c'''\delta|V(G)|$. Indeed, otherwise, by Lemma~\ref{L:tw-sep} we could find a $1/2$-separator of $G'$ of size at most $c'''\delta|V(G)|+1$, and by applying it again on at most two connected components and taking the union of the separators, for a small enough $c'''$ we would find a $1/4$-separator of $G'$ of size at most $c'' \delta|V(G)|$, a contradiction. Therefore $G'$ has treewidth $\Omega(\delta \cdot |V(G)|)$, which concludes the proof. 
\end{proof}

The previous considerations are used in the following lemma, which essentially provides a four-regular graph with $\Theta(t)$ vertices, genus at most~$g$, and treewidth $\Theta(\sqrt{gt})$.  Intuitively, if we start with an expander~$G$ of size~$g$ and set $\delta=\sqrt{t/g}$, then $G^\delta$ has $\Theta(t)$ vertices and treewidth $\Theta(\sqrt{gt})$.

\begin{lemma}\label{L:graphp}
  There is a universal constant $c_\tw$ such that, for every choice of $g\geq0$ and $t\geq 48(g+1)$, there is a four-regular graph~$P$ embeddable on a surface of genus~$g$ such that $|V(P)|\le t/12$ and $\tw(P)\ge\max\{2,c_\tw\sqrt{gt+t}\}$.
\end{lemma}
\begin{proof}
  Consider the family of expanders $\mathcal{H}$ defined in Lemma~\ref{L:existence} and let $\mu\ge1$ be the number of vertices of the smallest graph in that family. Then, by Lemma~\ref{L:existence}, for a small enough universal constant $c>0$, for any $g\ge\mu/c$, there exists an expander $H\in \mathcal{H}$ such that $c\cdot (g+1) \leq|V(H)|\leq g/2$.
  
  Consider $g$ and~$t$ as in the lemma; assume first that indeed $g\ge\mu/c$, and let $H$ be as indicated.  Since $H$ is four-regular, it has at most $g$ edges, so its genus is at most~$g$.  Let $\delta$ be $\lfloor \sqrt{t/(12|V(H)|)}\rfloor$.  Since $t\ge 48(g+1)\ge 96|V(H)|$, we have $\delta\ge2$ and $\delta=\Theta(\sqrt{t/(g+1)})$. Moreover, let $P:=H^\delta$ (this construction depends on a choice of a cellular embedding for $H$; we choose an arbitrary one on a surface of genus at most $g$).  Like~$H$, the graph~$P$ is embeddable on the surface of genus~$g$. By the choice of~$\delta$, the four-regular graph $P$ has $\Omega(t)$ vertices, but at most $t/12$ vertices.  The graph~$H$ has treewidth $\Omega(g+1)$ by Lemma~\ref{L:tw}, so Lemma~\ref{L:ght} implies that $P$ has treewidth $\Omega(\delta (g+1))=\Omega(\sqrt{gt+t})$.  Moreover, this treewidth is at least two, because $P$ contains a $\delta\times\delta$-grid.  This proves the lemma if $g\ge\mu/c$.

  Let now $g$ and~$t$ be as in the lemma, with $g<\mu/c$.  We let $H$ be the planar graph with one vertex and two loops, and let $\delta$ be equal to $\lfloor \sqrt{t/12}\rfloor$.  Note that $t\ge 48(g+1) \geq 48$; thus, $\delta\ge2$.  The graph $P := H^\delta$ has $\delta^2\le t/12$ vertices; it contains a $\delta\times\delta$-grid; hence, it has treewidth at least~$\delta$, which is both at least two and $\Omega(\sqrt{gt+t})$. Moreover, $P$ is planar and thus embeds on a surface of genus at most $g$.  This proves the lemma if $g<\mu/c$.
\end{proof}

Armed with Lemma~\ref{L:graphp}, we can now proceed to the proof of Proposition~\ref{P:multicut2}.
\begin{proof}[Proof of Proposition~\ref{P:multicut2}]
  The reduction is slightly different from the previous ones. It follows the same ideas, with the help of the previous lemmas, but in order to obtain exactly the lower bound, we will use directly the full strength of Theorem~\ref{T:beattreewidth} and not go through the \textsc{4-Regular Graph Tiling} problem.

  Let $P$ be a four-regular graph cellularly embedded on a surface~$\surf$ of genus at most~$g$.  We will make use of Theorem~\ref{T:beattreewidth} on CSPs whose primal graph is~$P$.  Given a binary CSP instance $I=(V,D,C)$ whose primal graph is~$P$, we can transform it in polynomial time into an equivalent instance $(G,T,\lambda)$ of \textsc{Multiway Cut} where $|T|=12|V(P)|$ and $G$ is cellularly embedded on~$\surf$ and has $|V(P)|\cdot\poly(|D|)$ vertices and edges, as follows.
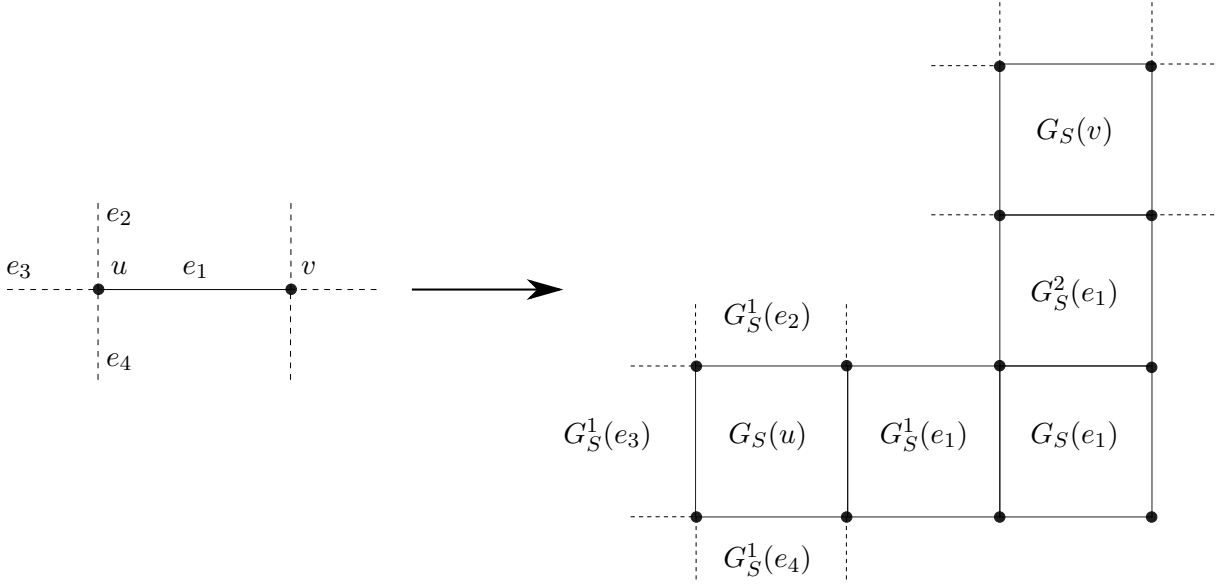
\begin{figure}
\centering
\def\svgwidth{\textwidth}
\input{chains.pdf_tex}
\caption{Reduction from a binary CSP instance to a \textsc{Multiway Cut} instance.  The figure displays the gadgets corresponding to an edge $e_1=uv$ of the binary CSP instance.  Note that $G_S^1(e_1)$ is attached to an arbitrary side of~$G_S(u)$, and similarly $G_S^2(e_1)$ is attached to an arbitrary side of~$G_S(v)$.  What matters is that the gadgets adjacent to $G_S^1(e_1)$ (resp.\ $G_S^2(e_1)$) are on opposite sides of this gadget; and similarly that the gadgets adjacent to $G_S(e_1)$ are on non-opposite sides of this gadget.}
\label{F:chains}
\end{figure}
\begin{enumerate}
\item For each vertex~$v$ of~$P$, we create a cross gadget $G_S(v)$ where $\Delta=|D|$ and the subset~$S$ is $\Diag([\Delta]):=\{(x,x)\mid x\in[\Delta]\}$.
\item\label{edges} For each edge $e=uv$ of~$P$, we create a cross gadget $G_S(e)$ where $\Delta=|D|$, as well as two cross gadgets $G^1_S(e)$ and $G^2_S(e)$.  The cross gadgets $G_S(u)$ and $G_S(v)$ are connected via the three gadgets $G^1_S(e)$, $G_S(e)$ and $G^2_S(e)$ as pictured in Figure~\ref{F:chains}, i.e., $G_S(u)$ has an arbitrary side identified to the $L$ side of $G^1_S(e)$, of which the $R$ side is identified to the $L$ side of $G_S(e)$, of which the $U$ side is identified to the $D$ side of $G^2_S(e)$, which is finally identified via its $U$ side to an arbitrary side of~$G_S(v)$.  This is done in such a way that the cyclic ordering of the four cross gadgets around a cross gadget $G_S(u)$ corresponds to the cyclic ordering of the four edges around $u$ specified by the embedding of $P$. This way, the new graph $G$ is cellularly embedded on the same surface~$\surf$. 
\item For each gadget $G^1_S(e)$ or $G^2_S(e)$, the corresponding set $S$ is the unconstrained relation $[\Delta]\times[\Delta]$.  (These gadgets are here only to separate the gadgets corresponding to the vertices of~$P$, see below.)  For $G_S(e)$, the corresponding set $S$ corresponds to the relation on~$e$, but we need to take care of the orientations.  Formally, let $\iota$ be the identity permutation on~$[\Delta]$ and~$\rho$ be the reversing permutation on~$[\Delta]$ (namely, $\rho(i)=\Delta+1-i$).  If $G^1_S(e)$ is attached to $G_S(u)$ via the $R$ or $U$ side of~$G_S(u)$, then let $\sigma_e:=\iota$, otherwise, let $\sigma_e:=\rho$.  Thus, $\sigma_e$ encodes the permutation between the indices of the side of $G_S(u)$ and the side of $G_S(e)$ which are connected through $G_S^1(e)$.  Similarly, if $G^2_S(e)$ is attached to $G_S(v)$ via the $D$ or $L$ side of~$G_S(v)$, then let $\tau_e:=\iota$; otherwise, let $\tau_e:=\rho$.  Finally, the set~$S$ corresponding to gadget $G_S(e)$ contains $(i,j)$ if and only if $(\sigma_e(i),\tau_e(j))$ is in the relation corresponding to~$e$.
\item\label{terminals} We place terminals at every corner of every cross gadget.  Let $T$ be this set of terminals; by construction, $|T|=4(|V(P)|+|E(P)|)=12|V(P)|$.
\item We let $\lambda:=7|V(P)|D_1$.
\end{enumerate}

Note that at the end of Step~\ref{edges}, each cross gadget for a vertex has its four sides identified to another gadget, while each of the three cross gadgets for an edge has two sides left unidentified. Therefore, after Step~\ref{terminals}, a corner of a cross gadget lies in either two or three distinct cross gadgets, so every cross gadget has at its corners four distinct terminals.

Let us first assume that the CSP instance is satisfiable.  For each vertex~$v$ of~$P$, one can use the value $s_v$ to choose, using Lemma~\ref{L:marx:gadgets}(1), a multiway cut in $G_S(v)$ representing $(s_v,s_v)$.  Similarly, in the gadgets adjacent to~$G_S(v)$, we choose a multiway cut representing $(s_v,s_v)$.  Finally, in the gadget $G_S(e)$ where $e=uv$, we choose a multiway cut representing $(\sigma_e(s_u),\tau_e(s_v))$.  By construction, taking the union of all these sets of edges forms a multiway cut separating all the terminals in~$G$. Indeed, after removing the multiway cuts, the four terminals of each cross gadget lie in four different components. This remains the case after identifying the sides: The binary constraints of the CSP force the multiway cuts to match along two consecutive sides.  The multiway cut has weight exactly~$\lambda$ since it is the union of $7|V(P)|$ edge sets of weight at most~$D_1$.

Let us prove the other direction.  First, if $(G,T)$ admits a multiway cut of weight at most $\lambda=7|V(P)|D_1$, then each of the cross gadgets $G_S$ must admit a multiway cut (otherwise the four terminals would not be disconnected). Since there are $7|V(P)|$ of those, and by Lemma~\ref{L:marx:gadgets}(2), each of these multiway cuts has weight exactly $D_1$. Therefore, by Lemma~\ref{L:marx:gadgets}(3), each of them has to represent some $(i,j) \in S$.  The value on the gadgets $G_S(v)$ are of the form $(s_v,s_v)$; we let $s_v$ be the value for the binary CSP instance for variable~$v$. Furthermore, the multiway cuts need to match along consecutive sides, otherwise, the four terminals are not separated in one of the cross gadgets, as pictured in Figure~\ref{F:Newclaim}. Therefore, the binary constraints of the CSP instance are satisfied.

To summarize: For some universal constant $d\ge1$, given a binary CSP instance $I=(V,D,C)$ whose primal graph is a four-regular graph~$P$ cellularly embedded on a surface of genus at most~$g$, we obtain in time $O((|V||D|)^d)$ an equivalent instance $(G,T,\lambda)$ of \textsc{Multiway Cut} where $|T|\le12|V|$ and $G$ has $O((|V||D|)^d)$ vertices and edges and is embeddable on a surface of genus~$g$.

Assume that, for some fixed~$g$, and for some fixed but large enough $t\ge48(g+1)$, there is an algorithm~$\mathbb A$ that decides all the \textsc{Multiway Cut} instances $(G,T,\lambda)$ for which $G$ is embeddable on a surface of genus~$g$ and such that $|T|\le t$ in time $O(n^{(\alpha_\CSP c_\tw/d)\cdot \sqrt{gt+t}/\log(g+t)})$.  Let $P$ be obtained as in Lemma~\ref{L:graphp}, and consider a binary CSP instance $I=(V,D,C)$ whose primal graph is~$P$.  Since $|V(P)|\le t/12$ and $t$ is fixed, we can in constant time compute a cellular embedding of~$P$ on a surface of genus at most~$g$.  Using algorithm~$\mathbb A$ and the above reduction, we can decide this instance in time $O(|D|^d)+O(|D|^d)^{(\alpha_\CSP c_\tw/d)\cdot\sqrt{gt+t}/\log(g+t)}$ (for fixed $g$ and~$t$).  Since also $\tw(P)\ge c_\tw\sqrt{gt+t}$ and $\log(g+t) \geq \log \tw(P)$ (because $P$ has at most $t/12$ vertices), this is $O(|D|^{\alpha_{\CSP} \cdot \tw(P)/\log \tw(P)})$ provided $t$ is large enough. Finally, by construction, $P$ has treewidth at least two.  Hence by Theorem~\ref{T:beattreewidth}, the Exponential Time Hypothesis does not hold.

So, there exists an integer~$\bar t$ such that the following holds: If for some fixed~$g$, and for some fixed but large enough $t\ge\bar t$ and $t\ge48(g+1)$, there is an algorithm that decides all the \textsc{Multiway Cut} instances $(G,T,\lambda)$ for which $G$ is embeddable on a surface of genus~$g$ and such that $|T|\le t$ in time $O(n^{(\alpha_\CSP c_\tw/d)\cdot \sqrt{gt+t}/\log(g+t)})$, then the ETH does not hold.

Let $\alpha_\MCtwo>0$ be at most~$\alpha_{\CSP}c_\textup{tw}/d$ and such that $\alpha_\MCtwo\cdot\sqrt{gt+t}/\log(g+t)<1$ for each $g$ and~$t$ such that $\bar t>t\ge48(g+1)$.  There is a trivial, unconditional $\Omega(n)$ lower bound for \textsc{Multiway Cut}.  Thus, for each $g$ and~$t$ such that $\bar t>t\ge48(g+1)$, there is no algorithm deciding the \textsc{Multiway Cut} instances on a surface of genus~$g$ with $t$ terminals in time $O(n^{(\alpha_\MCtwo \cdot \sqrt{gt+t}/\log(g+t))})$.  Finally, if for some $g$ and~$t$ such that $t\ge48(g+1)$, there is an algorithm deciding the \textsc{Multiway Cut} instances on a surface of genus~$g$ with $t$ terminals in time $O(n^{(\alpha_\MCtwo \cdot \sqrt{gt+t}/\log(g+t))})$, then the ETH does not hold.
\end{proof}

\section{Proof of Theorem~\ref{T:mainmulticut}}\label{S:finish}

Finally, the proof of Theorem~\ref{T:mainmulticut} is obtained by using Proposition~\ref{P:multicut1} or Proposition~\ref{P:multicut2}, depending on the tradeoff between $g$ and $t$:

\begin{proof}[Proof of Theorem~\ref{T:mainmulticut}]
  Let $\alpha_\MC>0$.  Let us assume that for some fixed choice of integers $g\geq0$ and $t\geq 4$, there is an algorithm~$\mathbb A$ that solves the \textsc{Multiway Cut} instances $(G,T,\lambda)$ for which $G$ is embeddable on the orientable surface of genus~$g$ and $|T| \leq t$ in time $O(n^{\alpha_{\MC} \sqrt{gt + g^2+t}/\log(g+t)})$.  We distinguish according to two cases for the choice of $(g,t)$.
  \begin{itemize}
  \item If $t<48(g+1)$, then $\sqrt{gt+g^2+t} \leq \sqrt{97(g+1)^2}$ and thus $\mathbb A$ solves, in particular, the \textsc{Multiway Cut} instances where $G$ is embeddable on the surface of genus~$g$ and $|T| = 4$, in time $O(n^{\alpha_{\MC}\cdot10(g+1)/\log (g+2)})$, so Proposition~\ref{P:multicut1} implies that the ETH does not hold if $\alpha_\MC\le\alpha_\MCone/10$.
  \item If $t\ge48(g+1)$, then $\sqrt{gt+g^2+t}/\log (g+t)\leq \sqrt{2(gt+t)}/\log(g+t)$, and thus Proposition~\ref{P:multicut2} implies that the ETH does not hold if $\alpha_{\MC}\leq \alpha_{\MCtwo}/\sqrt{2}$.
  \end{itemize}
  Thus, assuming that $\alpha_\MC>0$ is small enough, the existence of~$\mathbb A$ implies that the ETH does not hold, as desired.
\end{proof}

\paragraph*{Acknowledgements} We are grateful to the anonymous reviewers for a careful reading of our manuscript and helpful remarks. The first and the fourth authors are partially supported by the French ANR project ANR-18-CE40-0004-01 (FOCAL). The second and the fourth  authors are partially supported by the French
    ANR project ANR-17-CE40-0033 (SoS)
and the French ANR project ANR-19-CE40-0014 (MIN-MAX). The third author is supported by ERC Consolidator Grant
  SYSTEMATICGRAPH (No.~725978). The fourth author is
  partially supported by the ANR project {ANR-16-CE40-0009-01}
  (GATO) and the CNRS PEPS project COMP3D. Parts of this work were realized
  when he was working at GIPSA-lab in Grenoble.

\bibliographystyle{plainurl}
\bibliography{biblio}

\end{document}

%% file: cross1.pdf_tex
\begingroup%
  \makeatletter%
  \providecommand\color[2][]{%
    \errmessage{(Inkscape) Color is used for the text in Inkscape, but the package 'color.sty' is not loaded}%
    \renewcommand\color[2][]{}%
  }%
  \providecommand\transparent[1]{%
    \errmessage{(Inkscape) Transparency is used (non-zero) for the text in Inkscape, but the package 'transparent.sty' is not loaded}%
    \renewcommand\transparent[1]{}%
  }%
  \providecommand\rotatebox[2]{#2}%
  \newcommand*\fsize{\dimexpr\f@size pt\relax}%
  \newcommand*\lineheight[1]{\fontsize{\fsize}{#1\fsize}\selectfont}%
  \ifx\svgwidth\undefined%
    \setlength{\unitlength}{3119.30393391bp}%
    \ifx\svgscale\undefined%
      \relax%
    \else%
      \setlength{\unitlength}{\unitlength * \real{\svgscale}}%
    \fi%
  \else%
    \setlength{\unitlength}{\svgwidth}%
  \fi%
  \global\let\svgwidth\undefined%
  \global\let\svgscale\undefined%
  \makeatother%
  \begin{picture}(1,0.41427524)%
    \lineheight{1}%
    \setlength\tabcolsep{0pt}%
    \put(0,0){\includegraphics[width=\unitlength,page=1]{cross1.pdf}}%
    \put(0.02245369,0.40654084){\color[rgb]{0,0,0}\makebox(0,0)[t]{\lineheight{1.25}\smash{\begin{tabular}[t]{c}$UL$\end{tabular}}}}%
    \put(0.41994068,0.40654084){\color[rgb]{0,0,0}\makebox(0,0)[lt]{\lineheight{1.25}\smash{\begin{tabular}[t]{l}$UR$\end{tabular}}}}%
    \put(0.02360365,0.00184086){\color[rgb]{0,0,0}\makebox(0,0)[t]{\lineheight{1.25}\smash{\begin{tabular}[t]{c}$DL$\end{tabular}}}}%
    \put(0.43619382,0.00184086){\color[rgb]{0,0,0}\makebox(0,0)[t]{\lineheight{1.25}\smash{\begin{tabular}[t]{c}$DR$\end{tabular}}}}%
    \put(0,0){\includegraphics[width=\unitlength,page=2]{cross1.pdf}}%
    \put(0.11998224,0.40696818){\color[rgb]{0,0,0}\makebox(0,0)[t]{\lineheight{1.25}\smash{\begin{tabular}[t]{c}$u_1$\end{tabular}}}}%
    \put(0.19236028,0.40696818){\color[rgb]{0,0,0}\makebox(0,0)[t]{\lineheight{1.25}\smash{\begin{tabular}[t]{c}$u_2$\end{tabular}}}}%
    \put(0.26085225,0.40696818){\color[rgb]{0,0,0}\makebox(0,0)[t]{\lineheight{1.25}\smash{\begin{tabular}[t]{c}$u_3$\end{tabular}}}}%
    \put(0.33565907,0.40696818){\color[rgb]{0,0,0}\makebox(0,0)[t]{\lineheight{1.25}\smash{\begin{tabular}[t]{c}$u_4$\end{tabular}}}}%
    \put(0.01804226,0.31873996){\color[rgb]{0,0,0}\makebox(0,0)[t]{\lineheight{1.25}\smash{\begin{tabular}[t]{c}$\ell_1$\end{tabular}}}}%
    \put(0.01804226,0.24927648){\color[rgb]{0,0,0}\makebox(0,0)[t]{\lineheight{1.25}\smash{\begin{tabular}[t]{c}$\ell_2$\end{tabular}}}}%
    \put(0.01804226,0.17689844){\color[rgb]{0,0,0}\makebox(0,0)[t]{\lineheight{1.25}\smash{\begin{tabular}[t]{c}$\ell_3$\end{tabular}}}}%
    \put(0.01804226,0.10646344){\color[rgb]{0,0,0}\makebox(0,0)[t]{\lineheight{1.25}\smash{\begin{tabular}[t]{c}$\ell_4$\end{tabular}}}}%
    \put(0.11998224,0.0022682){\color[rgb]{0,0,0}\makebox(0,0)[t]{\lineheight{1.25}\smash{\begin{tabular}[t]{c}$d_1$\end{tabular}}}}%
    \put(0.18070206,0.00226822){\color[rgb]{0,0,0}\makebox(0,0)[t]{\lineheight{1.25}\smash{\begin{tabular}[t]{c}$d_2$\end{tabular}}}}%
    \put(0.26328104,0.0022682){\color[rgb]{0,0,0}\makebox(0,0)[t]{\lineheight{1.25}\smash{\begin{tabular}[t]{c}$d_3$\end{tabular}}}}%
    \put(0.33323028,0.0022682){\color[rgb]{0,0,0}\makebox(0,0)[t]{\lineheight{1.25}\smash{\begin{tabular}[t]{c}$d_4$\end{tabular}}}}%
    \put(0.43392601,0.10452041){\color[rgb]{0,0,0}\makebox(0,0)[t]{\lineheight{1.25}\smash{\begin{tabular}[t]{c}$r_4$\end{tabular}}}}%
    \put(0.433926,0.31388238){\color[rgb]{0,0,0}\makebox(0,0)[t]{\lineheight{1.25}\smash{\begin{tabular}[t]{c}$r_1$\end{tabular}}}}%
    \put(0.433926,0.24976223){\color[rgb]{0,0,0}\makebox(0,0)[t]{\lineheight{1.25}\smash{\begin{tabular}[t]{c}$r_2$\end{tabular}}}}%
    \put(0.433926,0.17689844){\color[rgb]{0,0,0}\makebox(0,0)[t]{\lineheight{1.25}\smash{\begin{tabular}[t]{c}$r_3$\end{tabular}}}}%
    \put(0,0){\includegraphics[width=\unitlength,page=3]{cross1.pdf}}%
    \put(0.56447022,0.40654084){\color[rgb]{0,0,0}\makebox(0,0)[t]{\lineheight{1.25}\smash{\begin{tabular}[t]{c}$UL^*$\end{tabular}}}}%
    \put(0.96195722,0.40654084){\color[rgb]{0,0,0}\makebox(0,0)[lt]{\lineheight{1.25}\smash{\begin{tabular}[t]{l}$UR^*$\end{tabular}}}}%
    \put(0.56562018,0.00184086){\color[rgb]{0,0,0}\makebox(0,0)[t]{\lineheight{1.25}\smash{\begin{tabular}[t]{c}$DL^*$\end{tabular}}}}%
    \put(0.97821038,0.00184086){\color[rgb]{0,0,0}\makebox(0,0)[t]{\lineheight{1.25}\smash{\begin{tabular}[t]{c}$DR^*$\end{tabular}}}}%
    \put(0.66199878,0.40696818){\color[rgb]{0,0,0}\makebox(0,0)[t]{\lineheight{1.25}\smash{\begin{tabular}[t]{c}$u_1^*$\end{tabular}}}}%
    \put(0.7343768,0.40696818){\color[rgb]{0,0,0}\makebox(0,0)[t]{\lineheight{1.25}\smash{\begin{tabular}[t]{c}$u_2^*$\end{tabular}}}}%
    \put(0.80286879,0.40696818){\color[rgb]{0,0,0}\makebox(0,0)[t]{\lineheight{1.25}\smash{\begin{tabular}[t]{c}$u_3^*$\end{tabular}}}}%
    \put(0.87767558,0.40696818){\color[rgb]{0,0,0}\makebox(0,0)[t]{\lineheight{1.25}\smash{\begin{tabular}[t]{c}$u_4^*$\end{tabular}}}}%
    \put(0.56005879,0.31873996){\color[rgb]{0,0,0}\makebox(0,0)[t]{\lineheight{1.25}\smash{\begin{tabular}[t]{c}$\ell_1^*$\end{tabular}}}}%
    \put(0.56005879,0.24927648){\color[rgb]{0,0,0}\makebox(0,0)[t]{\lineheight{1.25}\smash{\begin{tabular}[t]{c}$\ell_2^*$\end{tabular}}}}%
    \put(0.56005879,0.17689844){\color[rgb]{0,0,0}\makebox(0,0)[t]{\lineheight{1.25}\smash{\begin{tabular}[t]{c}$\ell_3^*$\end{tabular}}}}%
    \put(0.56005879,0.10646344){\color[rgb]{0,0,0}\makebox(0,0)[t]{\lineheight{1.25}\smash{\begin{tabular}[t]{c}$\ell_4^*$\end{tabular}}}}%
    \put(0.64813821,0.0022682){\color[rgb]{0,0,0}\makebox(0,0)[t]{\lineheight{1.25}\smash{\begin{tabular}[t]{c}$d_1^*$\end{tabular}}}}%
    \put(0.73065622,0.00226822){\color[rgb]{0,0,0}\makebox(0,0)[t]{\lineheight{1.25}\smash{\begin{tabular}[t]{c}$d_2^*$\end{tabular}}}}%
    \put(0.81317429,0.0022682){\color[rgb]{0,0,0}\makebox(0,0)[t]{\lineheight{1.25}\smash{\begin{tabular}[t]{c}$d_3^*$\end{tabular}}}}%
    \put(0.87723352,0.0022682){\color[rgb]{0,0,0}\makebox(0,0)[t]{\lineheight{1.25}\smash{\begin{tabular}[t]{c}$d_4^*$\end{tabular}}}}%
    \put(0.97594251,0.10452041){\color[rgb]{0,0,0}\makebox(0,0)[t]{\lineheight{1.25}\smash{\begin{tabular}[t]{c}$r_4^*$\end{tabular}}}}%
    \put(0.97594251,0.31388238){\color[rgb]{0,0,0}\makebox(0,0)[t]{\lineheight{1.25}\smash{\begin{tabular}[t]{c}$r_1^*$\end{tabular}}}}%
    \put(0.97594251,0.24976223){\color[rgb]{0,0,0}\makebox(0,0)[t]{\lineheight{1.25}\smash{\begin{tabular}[t]{c}$r_2^*$\end{tabular}}}}%
    \put(0.97594251,0.17689844){\color[rgb]{0,0,0}\makebox(0,0)[t]{\lineheight{1.25}\smash{\begin{tabular}[t]{c}$r_3^*$\end{tabular}}}}%
    \put(0,0){\includegraphics[width=\unitlength,page=4]{cross1.pdf}}%
  \end{picture}%
\endgroup%

%% file: GraphTiling.pdf_tex
\begingroup%
  \makeatletter%
  \providecommand\color[2][]{%
    \errmessage{(Inkscape) Color is used for the text in Inkscape, but the package 'color.sty' is not loaded}%
    \renewcommand\color[2][]{}%
  }%
  \providecommand\transparent[1]{%
    \errmessage{(Inkscape) Transparency is used (non-zero) for the text in Inkscape, but the package 'transparent.sty' is not loaded}%
    \renewcommand\transparent[1]{}%
  }%
  \providecommand\rotatebox[2]{#2}%
  \newcommand*\fsize{\dimexpr\f@size pt\relax}%
  \newcommand*\lineheight[1]{\fontsize{\fsize}{#1\fsize}\selectfont}%
  \ifx\svgwidth\undefined%
    \setlength{\unitlength}{4591.93837681bp}%
    \ifx\svgscale\undefined%
      \relax%
    \else%
      \setlength{\unitlength}{\unitlength * \real{\svgscale}}%
    \fi%
  \else%
    \setlength{\unitlength}{\svgwidth}%
  \fi%
  \global\let\svgwidth\undefined%
  \global\let\svgscale\undefined%
  \makeatother%
  \begin{picture}(1,0.35657779)%
    \lineheight{1}%
    \setlength\tabcolsep{0pt}%
    \put(0,0){\includegraphics[width=\unitlength,page=1]{GraphTiling.pdf}}%
    \put(0.01371076,0.02944504){\color[rgb]{0,0,0}\makebox(0,0)[t]{\lineheight{1.25}\smash{\begin{tabular}[t]{c}$\vdots$\end{tabular}}}}%
    \put(0,0){\includegraphics[width=\unitlength,page=2]{GraphTiling.pdf}}%
    \put(0.20113705,0.02944504){\color[rgb]{0,0,0}\makebox(0,0)[t]{\lineheight{1.25}\smash{\begin{tabular}[t]{c}$\vdots$\end{tabular}}}}%
    \put(0.0160901,0.18604689){\color[rgb]{0,0,0}\makebox(0,0)[t]{\lineheight{1.25}\smash{\begin{tabular}[t]{c}$u$\end{tabular}}}}%
    \put(0.21384737,0.25765165){\color[rgb]{0,0,0}\makebox(0,0)[t]{\lineheight{1.25}\smash{\begin{tabular}[t]{c}$v_a$\end{tabular}}}}%
    \put(0.21384737,0.20815527){\color[rgb]{0,0,0}\makebox(0,0)[t]{\lineheight{1.25}\smash{\begin{tabular}[t]{c}$v_b$\end{tabular}}}}%
    \put(0.21384737,0.16195865){\color[rgb]{0,0,0}\makebox(0,0)[t]{\lineheight{1.25}\smash{\begin{tabular}[t]{c}$v_c$\end{tabular}}}}%
    \put(0.21384737,0.11246227){\color[rgb]{0,0,0}\makebox(0,0)[t]{\lineheight{1.25}\smash{\begin{tabular}[t]{c}$v_d$\end{tabular}}}}%
    \put(0.41659929,0.24082288){\color[rgb]{0,0,0}\makebox(0,0)[t]{\lineheight{1.25}\smash{\begin{tabular}[t]{c}$u_1$\end{tabular}}}}%
    \put(0.41626931,0.18274713){\color[rgb]{0,0,0}\makebox(0,0)[t]{\lineheight{1.25}\smash{\begin{tabular}[t]{c}$u_2$\end{tabular}}}}%
    \put(0.41626931,0.12731119){\color[rgb]{0,0,0}\makebox(0,0)[t]{\lineheight{1.25}\smash{\begin{tabular}[t]{c}$u_3$\end{tabular}}}}%
    \put(0.41857915,0.06263593){\color[rgb]{0,0,0}\makebox(0,0)[t]{\lineheight{1.25}\smash{\begin{tabular}[t]{c}$u_4$\end{tabular}}}}%
    \put(0.09191991,0.20754429){\color[rgb]{0,0,0}\makebox(0,0)[t]{\lineheight{1.25}\smash{\begin{tabular}[t]{c}$e_1$\end{tabular}}}}%
    \put(0.11298965,0.18887804){\color[rgb]{0,0,0}\makebox(0,0)[t]{\lineheight{1.25}\smash{\begin{tabular}[t]{c}$e_2$\end{tabular}}}}%
    \put(0.13091488,0.16601188){\color[rgb]{0,0,0}\makebox(0,0)[t]{\lineheight{1.25}\smash{\begin{tabular}[t]{c}$e_3$\end{tabular}}}}%
    \put(0.1367891,0.13381259){\color[rgb]{0,0,0}\makebox(0,0)[t]{\lineheight{1.25}\smash{\begin{tabular}[t]{c}$e_4$\end{tabular}}}}%
    \put(0.53638221,0.26986075){\color[rgb]{0,0,0}\makebox(0,0)[lt]{\lineheight{1.25}\smash{\begin{tabular}[t]{l}$v_{e_1}$\end{tabular}}}}%
    \put(0,0){\includegraphics[width=\unitlength,page=3]{GraphTiling.pdf}}%
    \put(0.53638221,0.21541474){\color[rgb]{0,0,0}\makebox(0,0)[lt]{\lineheight{1.25}\smash{\begin{tabular}[t]{l}$v_{e_2}$\end{tabular}}}}%
    \put(0.53638221,0.15700901){\color[rgb]{0,0,0}\makebox(0,0)[lt]{\lineheight{1.25}\smash{\begin{tabular}[t]{l}$v_{e_3}$\end{tabular}}}}%
    \put(0.53638221,0.10256299){\color[rgb]{0,0,0}\makebox(0,0)[lt]{\lineheight{1.25}\smash{\begin{tabular}[t]{l}$v_{e_4}$\end{tabular}}}}%
    \put(0,0){\includegraphics[width=\unitlength,page=4]{GraphTiling.pdf}}%
    \put(0.7914201,0.32265688){\color[rgb]{0,0,0}\makebox(0,0)[t]{\lineheight{1.25}\smash{\begin{tabular}[t]{c}$v_{a,1}$\end{tabular}}}}%
    \put(0,0){\includegraphics[width=\unitlength,page=5]{GraphTiling.pdf}}%
    \put(0.83464692,0.22948975){\color[rgb]{0,0,0}\makebox(0,0)[t]{\lineheight{1.25}\smash{\begin{tabular}[t]{c}$v_{b,2}$\end{tabular}}}}%
    \put(0,0){\includegraphics[width=\unitlength,page=6]{GraphTiling.pdf}}%
    \put(0.92011071,0.13995235){\color[rgb]{0,0,0}\makebox(0,0)[t]{\lineheight{1.25}\smash{\begin{tabular}[t]{c}$v_{c,3}$\end{tabular}}}}%
    \put(0,0){\includegraphics[width=\unitlength,page=7]{GraphTiling.pdf}}%
    \put(0.98346606,0.01742922){\color[rgb]{0,0,0}\makebox(0,0)[t]{\lineheight{1.25}\smash{\begin{tabular}[t]{c}$v_{d,4}$\end{tabular}}}}%
    \put(0,0){\includegraphics[width=\unitlength,page=8]{GraphTiling.pdf}}%
    \put(0.27240153,0.2309236){\color[rgb]{0,0,0}\makebox(0,0)[lt]{\begin{minipage}{0.0013199\unitlength}\raggedright \end{minipage}}}%
  \end{picture}%
\endgroup%

%% file: Newclaim.pdf_tex
\begingroup%
  \makeatletter%
  \providecommand\color[2][]{%
    \errmessage{(Inkscape) Color is used for the text in Inkscape, but the package 'color.sty' is not loaded}%
    \renewcommand\color[2][]{}%
  }%
  \providecommand\transparent[1]{%
    \errmessage{(Inkscape) Transparency is used (non-zero) for the text in Inkscape, but the package 'transparent.sty' is not loaded}%
    \renewcommand\transparent[1]{}%
  }%
  \providecommand\rotatebox[2]{#2}%
  \newcommand*\fsize{\dimexpr\f@size pt\relax}%
  \newcommand*\lineheight[1]{\fontsize{\fsize}{#1\fsize}\selectfont}%
  \ifx\svgwidth\undefined%
    \setlength{\unitlength}{1379.44093878bp}%
    \ifx\svgscale\undefined%
      \relax%
    \else%
      \setlength{\unitlength}{\unitlength * \real{\svgscale}}%
    \fi%
  \else%
    \setlength{\unitlength}{\svgwidth}%
  \fi%
  \global\let\svgwidth\undefined%
  \global\let\svgscale\undefined%
  \makeatother%
  \begin{picture}(1,0.58328248)%
    \lineheight{1}%
    \setlength\tabcolsep{0pt}%
    \put(0,0){\includegraphics[width=\unitlength,page=1]{Newclaim.pdf}}%
    \put(-0.00149517,0.56565759){\color[rgb]{0,0,0}\makebox(0,0)[lt]{\lineheight{0}\smash{\begin{tabular}[t]{l}$UL$\end{tabular}}}}%
    \put(0.48234471,0.56565759){\color[rgb]{0,0,0}\makebox(0,0)[lt]{\lineheight{0}\smash{\begin{tabular}[t]{l}$UR$\end{tabular}}}}%
    \put(0.48191428,0.00340944){\color[rgb]{0,0,0}\makebox(0,0)[lt]{\lineheight{0}\smash{\begin{tabular}[t]{l}$DR$\end{tabular}}}}%
    \put(-0.0019256,0.00340944){\color[rgb]{0,0,0}\makebox(0,0)[lt]{\lineheight{0}\smash{\begin{tabular}[t]{l}$DL$\end{tabular}}}}%
    \put(0,0){\includegraphics[width=\unitlength,page=2]{Newclaim.pdf}}%
    \put(0.93217623,0.56565759){\color[rgb]{0,0,0}\makebox(0,0)[lt]{\lineheight{0}\smash{\begin{tabular}[t]{l}$UL$\end{tabular}}}}%
    \put(0.9416317,0.00340944){\color[rgb]{0,0,0}\makebox(0,0)[lt]{\lineheight{0}\smash{\begin{tabular}[t]{l}$DL$\end{tabular}}}}%
  \end{picture}%
\endgroup%

%% file: blowup.pdf_tex
\begingroup%
  \makeatletter%
  \providecommand\color[2][]{%
    \errmessage{(Inkscape) Color is used for the text in Inkscape, but the package 'color.sty' is not loaded}%
    \renewcommand\color[2][]{}%
  }%
  \providecommand\transparent[1]{%
    \errmessage{(Inkscape) Transparency is used (non-zero) for the text in Inkscape, but the package 'transparent.sty' is not loaded}%
    \renewcommand\transparent[1]{}%
  }%
  \providecommand\rotatebox[2]{#2}%
  \newcommand*\fsize{\dimexpr\f@size pt\relax}%
  \newcommand*\lineheight[1]{\fontsize{\fsize}{#1\fsize}\selectfont}%
  \ifx\svgwidth\undefined%
    \setlength{\unitlength}{7491.72747402bp}%
    \ifx\svgscale\undefined%
      \relax%
    \else%
      \setlength{\unitlength}{\unitlength * \real{\svgscale}}%
    \fi%
  \else%
    \setlength{\unitlength}{\svgwidth}%
  \fi%
  \global\let\svgwidth\undefined%
  \global\let\svgscale\undefined%
  \makeatother%
  \begin{picture}(1,0.20982594)%
    \lineheight{1}%
    \setlength\tabcolsep{0pt}%
    \put(0,0){\includegraphics[width=\unitlength,page=1]{blowup.pdf}}%
  \end{picture}%
\endgroup%

%% file: chains.pdf_tex
\begingroup%
  \makeatletter%
  \providecommand\color[2][]{%
    \errmessage{(Inkscape) Color is used for the text in Inkscape, but the package 'color.sty' is not loaded}%
    \renewcommand\color[2][]{}%
  }%
  \providecommand\transparent[1]{%
    \errmessage{(Inkscape) Transparency is used (non-zero) for the text in Inkscape, but the package 'transparent.sty' is not loaded}%
    \renewcommand\transparent[1]{}%
  }%
  \providecommand\rotatebox[2]{#2}%
  \newcommand*\fsize{\dimexpr\f@size pt\relax}%
  \newcommand*\lineheight[1]{\fontsize{\fsize}{#1\fsize}\selectfont}%
  \ifx\svgwidth\undefined%
    \setlength{\unitlength}{6611.62685176bp}%
    \ifx\svgscale\undefined%
      \relax%
    \else%
      \setlength{\unitlength}{\unitlength * \real{\svgscale}}%
    \fi%
  \else%
    \setlength{\unitlength}{\svgwidth}%
  \fi%
  \global\let\svgwidth\undefined%
  \global\let\svgscale\undefined%
  \makeatother%
  \begin{picture}(1,0.48298797)%
    \lineheight{1}%
    \setlength\tabcolsep{0pt}%
    \put(0,0){\includegraphics[width=\unitlength,page=1]{chains.pdf}}%
    \put(0.09283296,0.2543187){\color[rgb]{0,0,0}\makebox(0,0)[t]{\lineheight{1.25}\smash{\begin{tabular}[t]{c}$u$\end{tabular}}}}%
    \put(0.24931612,0.2543187){\color[rgb]{0,0,0}\makebox(0,0)[t]{\lineheight{1.25}\smash{\begin{tabular}[t]{c}$v$\end{tabular}}}}%
    \put(0.15587125,0.25452031){\color[rgb]{0,0,0}\makebox(0,0)[t]{\lineheight{1.25}\smash{\begin{tabular}[t]{c}$e_1$\end{tabular}}}}%
    \put(0.09283296,0.2984764){\color[rgb]{0,0,0}\makebox(0,0)[t]{\lineheight{1.25}\smash{\begin{tabular}[t]{c}$e_2$\end{tabular}}}}%
    \put(0.01002427,0.25452032){\color[rgb]{0,0,0}\makebox(0,0)[t]{\lineheight{1.25}\smash{\begin{tabular}[t]{c}$e_3$\end{tabular}}}}%
    \put(0.09283301,0.17790957){\color[rgb]{0,0,0}\makebox(0,0)[t]{\lineheight{1.25}\smash{\begin{tabular}[t]{c}$e_4$\end{tabular}}}}%
    \put(0,0){\includegraphics[width=\unitlength,page=2]{chains.pdf}}%
    \put(0.62877345,0.11378127){\color[rgb]{0,0,0}\makebox(0,0)[t]{\lineheight{1.25}\smash{\begin{tabular}[t]{c}$G_S(u)$\end{tabular}}}}%
    \put(0.75776699,0.11378127){\color[rgb]{0,0,0}\makebox(0,0)[t]{\lineheight{1.25}\smash{\begin{tabular}[t]{c}$G^1_S(e_1)$\end{tabular}}}}%
    \put(0.88243895,0.11378127){\color[rgb]{0,0,0}\makebox(0,0)[t]{\lineheight{1.25}\smash{\begin{tabular}[t]{c}$G_S(e_1)$\end{tabular}}}}%
    \put(0.88243895,0.23194136){\color[rgb]{0,0,0}\makebox(0,0)[t]{\lineheight{1.25}\smash{\begin{tabular}[t]{c}$G^2_S(e_1)$\end{tabular}}}}%
    \put(0.88243895,0.36715545){\color[rgb]{0,0,0}\makebox(0,0)[t]{\lineheight{1.25}\smash{\begin{tabular}[t]{c}$G_S(v)$\end{tabular}}}}%
    \put(0,0){\includegraphics[width=\unitlength,page=3]{chains.pdf}}%
    \put(0.62877345,0.21295721){\color[rgb]{0,0,0}\makebox(0,0)[t]{\lineheight{1.25}\smash{\begin{tabular}[t]{c}$G^1_S(e_2)$\end{tabular}}}}%
    \put(0.62877345,0.01006787){\color[rgb]{0,0,0}\makebox(0,0)[t]{\lineheight{1.25}\smash{\begin{tabular}[t]{c}$G^1_S(e_4)$\end{tabular}}}}%
    \put(0.49653888,0.11378127){\color[rgb]{0,0,0}\makebox(0,0)[t]{\lineheight{1.25}\smash{\begin{tabular}[t]{c}$G^1_S(e_3)$\end{tabular}}}}%
    \put(0,0){\includegraphics[width=\unitlength,page=4]{chains.pdf}}%
  \end{picture}%
\endgroup%